\documentclass[%
 reprint, 
superscriptaddress,
 amsmath,amssymb,
 aps, prl, 
]{revtex4-2}

\usepackage{graphicx} 
\usepackage{dcolumn}
\usepackage{bm}
\usepackage{amsthm}
\usepackage[colorlinks=true]{hyperref}
\usepackage{mathtools}
\usepackage{dsfont}
\usepackage{color}
\usepackage{subfig}
\usepackage{comment}

\newtheorem{definition}{Definition}

\newtheorem{lemma}{Lemma}
\newtheorem{theorem}{Theorem}
\newtheorem{proposition}[theorem]{Proposition}

\newtheorem*{remark}{Remark}
\newtheorem*{protocol}{Protocol}

\newcommand{\mbb}{\mathbb}
\newcommand{\mc}{\mathcal}
\newcommand{\N}{\mathcal{N}}
\newcommand{\M}{\mathcal{M}}

\newcommand{\tr}{\textrm{Tr}}

\newcommand{\ket}[1]{|#1\rangle}
\newcommand{\bra}[1]{\langle #1|}
\newcommand{\op}[2]{|#1\rangle\langle #2|}

\newcommand{\1}{\mathds{1}}

\newcommand{\wt}{\widetilde}
\newcommand{\state}[1]{\ket{#1}\bra{#1}}

\definecolor{cool_green}{rgb}{0.0, 0.5, 0.0}

\begin{document}
\title{Universal Improvement of Channel Fidelity via Entanglement Assistance}
\author{Xinan Chen}\email{xinan.chen.137@gmail.com}
\affiliation{Department of Electrical and Computer Engineering, University of Illinois at Urbana-Champaign, Urbana, Illinois 61801, USA}
\author{Eric Chitambar}\email{echitamb@illinois.edu}
\affiliation{Department of Electrical and Computer Engineering, University of Illinois at Urbana-Champaign, Urbana, Illinois 61801, USA}
\date{\today}

\begin{abstract}
    Communication using quantum channels generally requires encoding and decoding on many identical uses of a quantum channel. In practical settings, decoherence may severely limit our ability to do so, potentially rendering each channel use nonidentical. Given $n$ nonidentical channels, we present an entanglement-assisted encoding and decoding strategy that yields a channel with higher worst-case fidelity than all of the $n$ given channels. The strategy is universal, in the sense that the improvement holds regardless of what channels are given, as long as they satisfy mild assumptions on their initial entanglement fidelities. Our protocol is optimal when the given channels are depolarizing channels, and is essentially the unique one that is universal. This idea can also be extended to classical channels, where shared randomness between the sender and the receiver allows universal enhancement of the probability of correct transmission. We then consider general quantum resource theories, providing a necessary condition and a sufficient condition for the possibility of such universal improvement of quantum resources.
\end{abstract}

\maketitle

\section{Introduction}
Communication using quantum channels is a central topic in quantum information theory. Following the conventional communication framework due to Claude Shannon \cite{shannon1948mathematical}, classical or quantum information is transmitted from a sender to a receiver by many uses of a quantum channel. To transmit a quantum state $\psi$, the sender encodes $\psi$ into $n$ registers, which are sent to the receiver by using the channel $n$ times. The channel incurs independent and identically distributed (i.i.d.) noise on each of the registers. Upon receiving the channel output, the receiver performs a decoding operation on the $n$ received registers, such that the decoded quantum state $\psi'$ is close to $\psi$. Moreover, this problem is usually studied in the asymptotic setting where the number of channel uses tends to infinity. The notion of quantum capacity of a channel, for instance, is defined precisely in this limit \cite{lloyd1997capacity,shor2002quantum,devetak2005private}. 

In practical settings, however, many factors, such as decoherence and experimental imperfections, will break the i.i.d. assumption. At the same time, arbitrary manipulation of large quantum systems is also unattainable in the near future. Therefore, quantum communication beyond the asymptotic and i.i.d. regime has garnered significant interest within the field. For example, Refs. \cite{buscemi2010quantum,beigi2016decoding,tomamichel2016quantum,wang2018semidefinite} gave bounds on the one-shot quantum capacity of a quantum channel, i.e., the amount of quantum information that can be transmitted through a single use of a quantum channel up to $\epsilon$ error. Along the same lines, Ref. \cite{datta2013one} assumes the sender and receiver are allowed to share entanglement resources and bounds the entanglement-assisted one-shot quantum capacity. 

All of these works require us to have precise knowledge of the quantum channels we intend to use. However, acquiring this knowledge by process tomography \cite{chuang1997prescription} can be prohibitive, since estimating a channel up to $\epsilon$ error requires at least  $\Omega(d_{in}^2d_{out}^2/\epsilon^2)$ queries \cite{haah2016sample,mele2025optimal,chen2026quantum}, where $d_{in}$ and $d_{out}$ are the dimensions of the input and output systems. Furthermore, environmental factors can affect the action of a channel in ways that are difficult to account for. For example, polarization drift in optical fibers is time-dependent and highly susceptible to factors such as temperature, wind speed, and physical stress \cite{neumann2022continuous,bersin2024development,craddock2024automated,kapoor2025public}. As such, in the present work, we ask an alternative question: Given $n$ potentially distinct quantum channels, is there a universal, one-size-fits-all method to extract a channel that has higher quality than all $n$ given channels? If so, we can blindly apply this method to any set of quantum channels without needing any information about the channels, regardless of how much they can vary. Similar questions have recently been raised in other contexts. In Ref. \cite{zang2025no}, the authors investigated whether a universal entanglement purification protocol exists and proved a no-go theorem for two-qubit entanglement. Later, it was shown that universal state purification cannot be achieved using only classically simulable operations \cite{he2025no}, while Ref. \cite{liu2025no} showed that universal state purification is impossible if a pure state output is required.

The quality of a quantum channel $\mc{N}$ with input system $A_0$ and output system $A_1$ such that $A_0\cong A_1\cong \mathbb{C}^d$ can be measured by its worst-case fidelity, defined as 
\begin{align}
    F(\mc{N}) = \inf_{\psi^{RA_0}}\ F\left(\mathrm{id}^R\otimes\mc{N}^{A_0\to A_1}(\psi^{RA_0}), \psi^{RA_1}\right)
\end{align}
where $F(\rho,\sigma)=\|\sqrt{\rho}\sqrt{\sigma}\|_1^2$ is the fidelity between $\rho$ and $\sigma$, and the infimum is over all reference systems $R$ and all pure states $\psi^{RA_0}$. We will henceforth simply call this quantity the fidelity of a quantum channel. In our work, we demonstrate that universal improvement of channel fidelity is possible by direct construction. To be more precise, given $n$ unknown and arbitrarily given quantum channels $\{\mc{N}_i\}_{i=1}^{n}$ with $F(\mc{N}_i)>1/d^2$, our entanglement-assisted protocol, by jointly acting on the $n$ channels, probabilistically produces a channel $\mc{N}_{\mathrm{out}}$ whose worst-case fidelity $F(\mc{N}_\mathrm{out}) > F(\mc{N}_i)$ for all $i$. We emphasize that our goal differs from the standard channel coding problem. Our main focus is that our protocol works for arbitrarily given $\{\N_i\}_{i=1}^n$ as long as they satisfy a weak fidelity constraint. 

Our approach is resource-theoretic at its core. In what follows, our problem will be phrased in the language of the resource theory of communication \cite{takagi2020application}. As a resource theory of quantum channels (i.e., a dynamical resource theory \cite{gour2019comparison,gour2020dynamicalentanglement,gour2020dynamicalresources}), it is defined by a set of free channels $\mbb{F}_{C}$ and a set of free superchannels $\mbb{F}_{SC}$. We will use $A_0, A_1, B_0, B_1$, etc. to denote quantum systems in which the subscript $0$ or $1$ indicates whether the system is the input or output of some quantum channel, respectively. A dynamical system $A_0\to A_1$ with input system $A_0$ and output system $A_1$ will be denoted simply by $A$. A superchannel $\Theta^{A\to B}$ is a completely CPTP-preserving linear map transforming channels on one dynamical system $A$ to another $B$. The action of any superchannel admits the physical decomposition \cite{chiribella2008transforming,gour2019comparison}
\begin{align}
    \Theta^{A\rightarrow B}[\mc{N}^A]=\mc{F}^{EA_1\rightarrow B_1}\circ\mc{N}^A\circ\mc{E}^{B_0\rightarrow A_0E} \label{eq:SC-structure}
\end{align}
where $\mc{E}$ and $\mc{F}$ are CPTP maps. Thus, a superchannel indeed represents a physical transformation of quantum channels, with a preprocessing (or an encoding operation) given by $\mc{E}$ and a postprocessing (or a decoding operation) given by $\mc{F}$. This is illustrated in Fig. \ref{fig:superchannel}. 

In our paper, we will consider probabilistic transformations of quantum channels, which are termed quantum superinstruments \cite{gour2020dynamicalresources,burniston2020necessary}. Recall that a probabilistic transformation of quantum states is modeled by a quantum instrument $\{\mc{N}_0,\mc{N}_1\}$, where $\mc{N}_0$ and $\mc{N}_1$ are completely positive and trace-non-increasing (CPTNI) maps, and $\mc{N}_0+\mc{N}_1$ is trace-preserving (TP). The quantum instrument acts on a quantum state $\rho$ by $\mc{N}_0(\rho)\otimes\state{0}+\mc{N}_1(\rho)\otimes\state{1}$, where classical registers $\state{0}$ and $\state{1}$ signal if the transformation succeeded or failed. Likewise, a probabilistic transformation of quantum channels is modeled as a quantum superinstrument $\{\Theta_0,\Theta_1\}$ where $\Theta_0$ and $\Theta_1$ are completely CPTNI-preserving supermaps, or equivalently, $\Theta_i^{A\to B}[\mc{N}^{RA}]$ is a CPTNI map whenever $\mc{N}^{RA}$ is a CPTP map. Furthermore, $\Theta_0+\Theta_1$ is required to be completely CPTP-preserving. The action of the superinstrument $\{\Theta_0,\Theta_1\}$ on a channel $\mc{N}$ is given by $\Theta_0[\mc{N}]\otimes\state{0}+\Theta_1[\mc{N}]\otimes\state{1}$. Similar to a superchannel, the action of a superinstrument $\{\Theta_x\}_x$ admits the physical decomposition
\begin{align}
    \Theta_x[\mc{N}] = \mc{F}^{EA_1\to B_1}_x \circ \mc{N}^A \circ \mc{E}^{B_0\to A_0E} 
\end{align}
where $\mc{E}$ is a CPTP map, while $\{\mc{F}_x\}_x$ is a quantum instrument (see Fig. \ref{fig:superinstrument}).

\begin{figure}
    \centering
    \subfloat[]{
    \includegraphics[width=\linewidth]{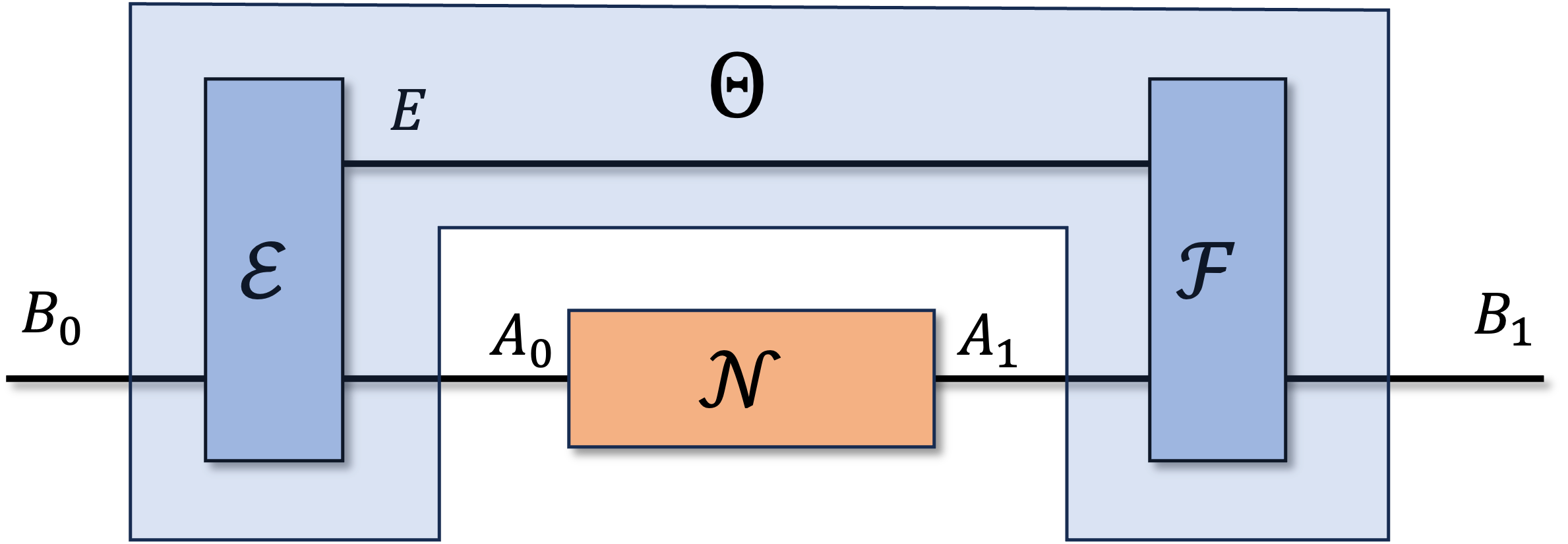}
    \label{fig:superchannel}
    } \\
    \subfloat[]{
    \includegraphics[width=\linewidth]{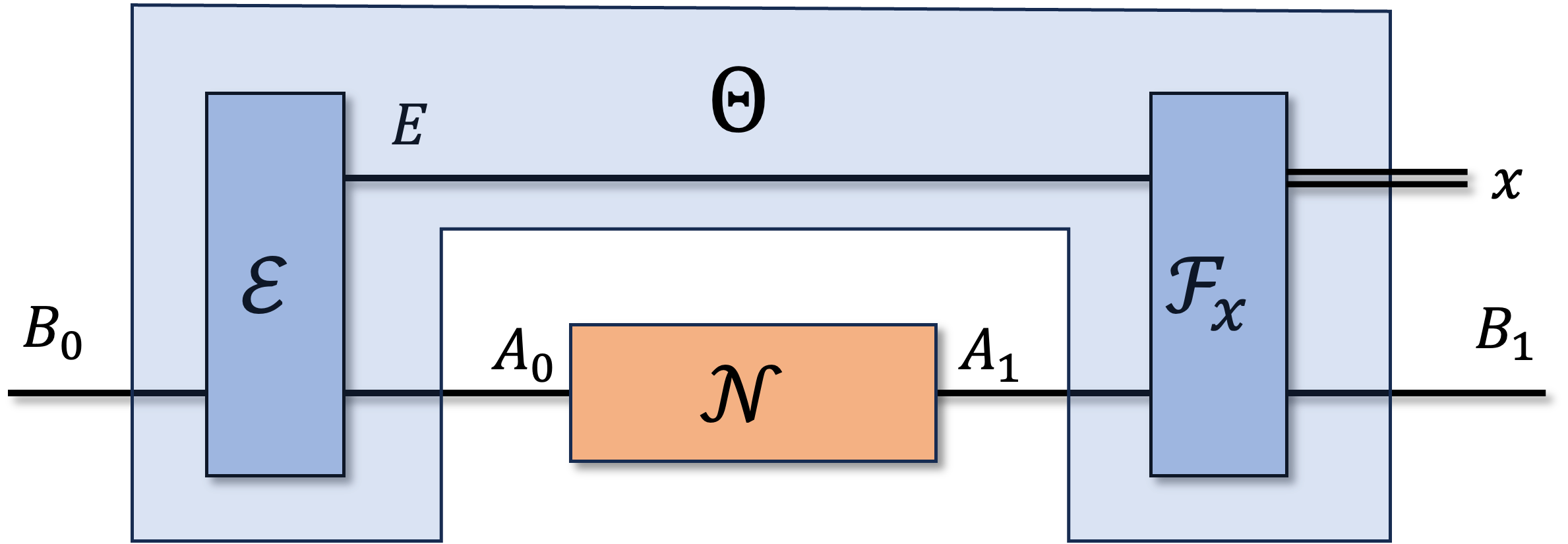}
    \label{fig:superinstrument}
    }
    \caption{A diagrammatic representation of (a) a superchannel and (b) a superinstrument.}
\end{figure}

In a communication scenario, the separation between the sender and the receiver restricts the superchannels and superinstruments they can apply to a quantum channel. This setting is formalized by the resource theory of communication~\cite{takagi2020application}, in which the two parties cannot communicate on their own. Instead, a quantum channel connecting them is provided as a resource, and they may only pre- and post-process it in a non-communicative fashion. A natural class of superchannels that encapsulates such operations is that of no-signaling (NS) superchannels, previously considered in Refs.~\cite{duan2016no,leung2015power,wang2018semidefinite,ji2024postselected}. Formally, a superchannel $\Theta^{A\to B}$ can be regarded as a bipartite channel from $B_0A_1$ to $A_0B_1$. We say that $\Theta$ is NS if~\cite{duan2016no}
\begin{align}
    \tr_{A_0}\left[\Theta\left(\rho^{A_1}\otimes\sigma^{B_0}\right)\right] &= \tr_{A_0}\left[\Theta\left(\rho^{A_1}\otimes\sigma'^{B_0}\right)\right], \\
    \tr_{B_1}\left[\Theta\left(\rho^{A_1}\otimes\sigma^{B_0}\right)\right] &= \tr_{B_1}\left[\Theta\left(\rho'^{A_1}\otimes\sigma^{B_0}\right)\right]
\end{align}
for all states $\rho,\rho'$ on $A_1$ and $\sigma,\sigma'$ on $B_0$. That is, the reduced output on the receiver's side ($B_1$) depends only on the receiver's input ($A_1$), and the reduced output on the sender's side ($A_0$) depends only on the sender's input ($B_0$); hence, neither party can signal to the other through $\Theta$. Under these constraints, the only channels the parties can implement without a resource channel are replacement channels $\mathcal{R}_\omega(\cdot)=\tr(\cdot)\,\omega$, which output a fixed state $\omega$ regardless of the input; these are therefore the free channels of the theory. Moreover, a superchannel is NS if and only if it maps every replacement channel to a replacement channel~\cite{takagi2020application}. In other words, the NS superchannels are exactly the resource-non-generating (RNG) superchannels of this theory. 

In this language, our goal is to find an appropriate NS superinstrument $\Theta=\{\Theta_0,\Theta_1\}$ such that given an arbitrary state $\psi$ that the sender wishes to send to the receiver and $n$ channels $\{\mc{N}_i\}_{i=1}^n$ arbitrarily chosen from some set $S$, the output state conditioned on success
\begin{align}
    \rho_{out}^{RB_1} = \frac{\Theta_0[\mc{N}_1\otimes\cdots\otimes\mc{N}_n](\psi^{RB_0})}{\tr\Theta_0[\mc{N}_1\otimes\cdots\otimes\mc{N}_n](\psi^{RB_0})}
\end{align}
has higher fidelity with respect to the initial state $\psi$ than the fidelity of the given channels $F(\mc{N}_i)$ (see Fig. \ref{fig:output-state} for an illustration). In other words, 
\begin{align}
    \inf_\psi F(\rho_{out}^{RB_1},\psi^{RB_0}) > F(\mc{N}_i) \quad\forall i.
\end{align}
Importantly, we reiterate that the superinstrument $\Theta$ does not allow for any additional communication; all communication from the sender to the receiver must come from the channels $\mc{N}_i$. However, they are allowed arbitrary quantum resources that do not convey information, including but not limited to quantum entanglement. It turns out that we can without loss of generality choose $\Theta_0$ such that the protocol succeeds with a probability that is fixed given the $n$ channels and is independent of the input state $\rho$ (see the Supplementary Material for details). That is,
\begin{align}
    \tr\Theta_0[\mc{N}_1\otimes\cdots\otimes\mc{N}_n](\rho)=p(\Theta_0;\mc{N}_1,\cdots,\mc{N}_n) \quad \forall\,\rho.
\end{align}
In this case, an output channel is well-defined and is given by 
\begin{align}\label{Eq:output-channel}
    \mc{N}_{out} = \frac{\Theta_0[\mc{N}_1\otimes\cdots\otimes\mc{N}_n]}{p(\Theta_0;\mc{N}_1,\cdots,\mc{N}_n)}.
\end{align}
Then, we can interpret our goal to be the task of universal channel distillation such that 
\begin{align}\label{eq:universal-def}
    F(\mc{N}_\mathrm{out}) > \max_i F(\mc{N}_i),
\end{align}
for arbitrary $\mc{N}_i\in S$. 
\begin{figure}
    \centering
    \includegraphics[width=\linewidth]{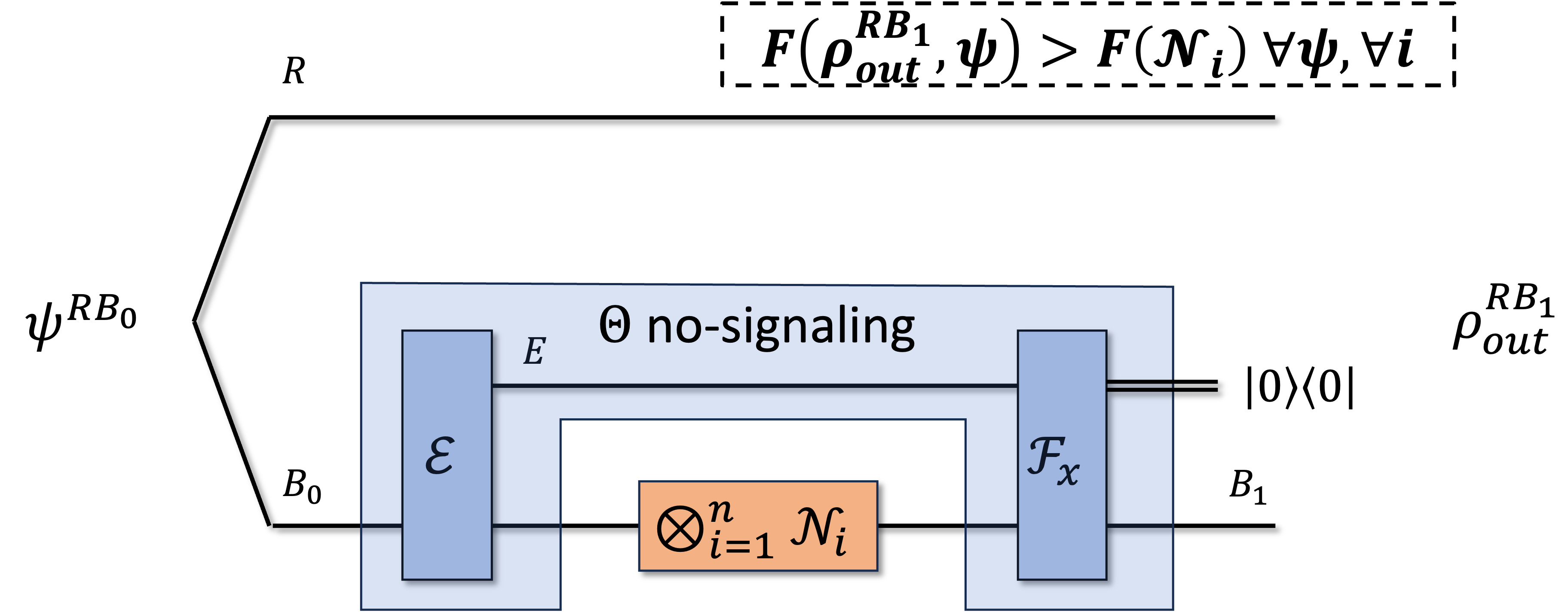}
    \caption{The task of universal channel fidelity improvement.}
    \label{fig:output-state}
\end{figure}

Even though, as we mentioned previously, strong constraints have been demonstrated for universal entanglement distillation and universal state purification, we show that such universal distillation is possible in the context of the resource theory of communication by constructing an explicit entanglement-assisted protocol. We will show that this protocol in fact achieves the optimal fidelity and success probability when the given channels are depolarizing channels. That is, even if we know precisely the depolarizing parameters for these channels, we cannot achieve a higher fidelity or succeed with higher probability in doing so. Moreover, our construction is essentially unique: all other possible constructions must be equivalent to our construction up to twirling the input and output channels and rescaling success probabilities. Our entanglement-assisted protocol can also be adapted to universally enhance the correct transmission probability of fully classical channels. Finally, we will generalize our construction to abstract resource theories, as we provide conditions for the existence of a universal resource distillation protocol.

\section{Results}

\subsection{The entanglement-assisted protocol}
\begin{figure*}
    \centering
    \includegraphics[width=0.8\linewidth]{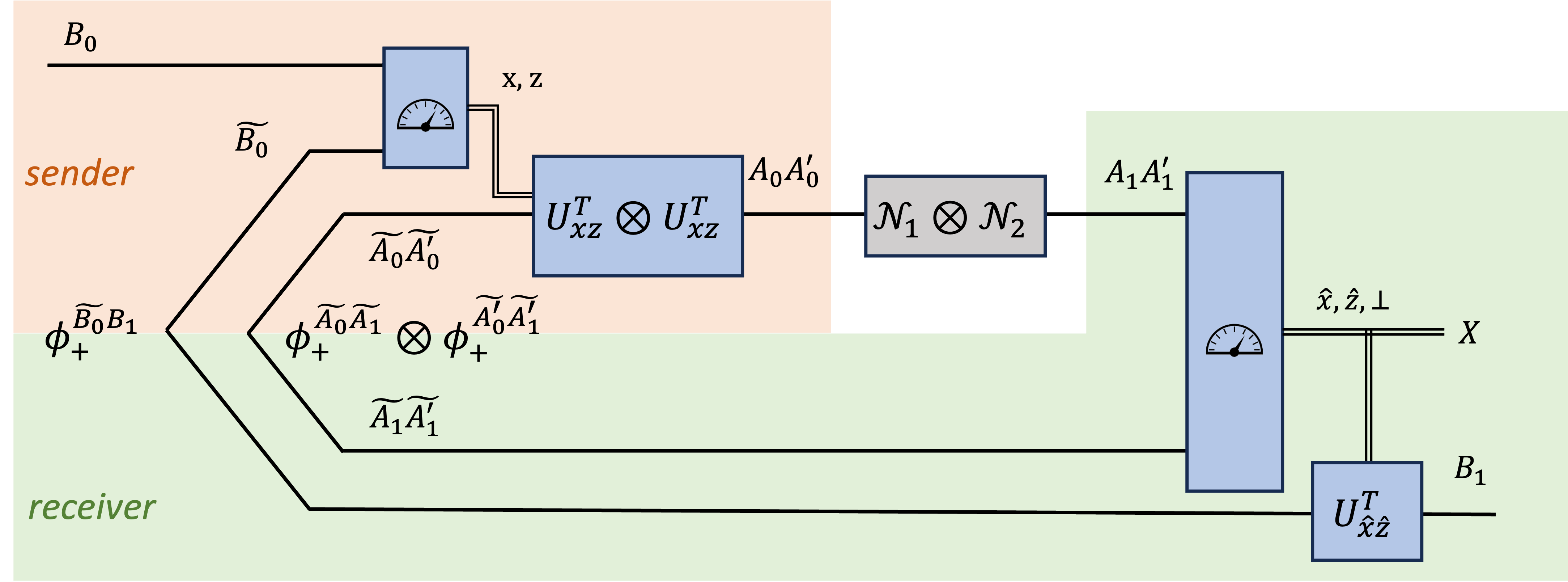}
    \caption{Constructing entanglement-assisted superinstrument using a teleportation-based protocol.}
    \label{fig:protocol}
\end{figure*}

Our first result is the existence of a protocol that universally improves the worst-case fidelity of the given channels, formally summarized by the following theorem. First, recall that the entanglement fidelity of a quantum channel $\mc{N}$ is defined as \cite{horodecki1999general}
\begin{align}
    F_e(\mc{N}) = F\left(\mathrm{id}^R\otimes\mc{N}^A(\phi_+^{RA_0}), \phi_+^{RA_1}\right),
\end{align}
where $\phi_+=\state{\phi_+}$ with $\ket{\phi_+} = \frac{1}{\sqrt{d}}\sum_{i=0}^{d-1}\ket{ii}$ is the $d$-dimensional maximally entangled state. Clearly, $F(\mc{N})\leq F_e(\mc{N})$. 
\begin{theorem}\label{thm:ent-assisted-performance}
    There exists an entanglement-assisted superinstrument $\{\Theta_0,\Theta_1\}$ such that the output channel $\mc{N}_{out} = \Theta_0[\mc{N}_1\otimes\cdots\otimes\mc{N}_n]/p(\Theta_0;\mc{N}_1,\cdots,\mc{N}_n)$ has worst-case fidelity
    \begin{align}
        F(\mc{N}_{out}) &= \frac{\prod_i F_e(\N_i)}{\prod_i F_e(\N_i) + \left(\frac{1}{d^2-1}\right)^{n-1}\prod_i(1-F_e(\N_i))}.
    \end{align}
    Moreover, $F(\N_{out}) > F_e(\N_i) \geq F(\N_i)$ for all $i$ whenever all $F_e(\N_i)$ satisfy $1/d^2<F_e(\N_i)<1$.
\end{theorem}

Before stating the protocol, let us recall the concept of generalized Bell states with local dimension $d$. Consider the shift and clock operators
\begin{align}
    X = \sum_{k=0}^{d-1} \ket{k\oplus1}\bra{k}, \qquad Z = \sum_{k=0}^{d-1} e^{i2\pi k/d} \state{k},
\end{align}
where $\oplus$ denotes addition modulo $d$. The set of generalized Pauli operators is $\{U_{xz}=X^xZ^z: x,z\in\{0,1,\cdots,d-1\}\}$. Now, the generalized Bell states are 
\begin{align}
    \ket{\phi_{xz}} = \mathbb{I}\otimes U_{xz}\ket{\phi_+},
\end{align}
and they form an orthonormal basis for $\mathbb{C}^d\otimes\mathbb{C}^d$. With this in mind, let us present the protocol.

\begin{protocol}
The entanglement-assisted protocol proceeds as follows (see Fig. \ref{fig:protocol} for an illustration).
\begin{enumerate}
    \item \textit{Shared resource.} The sender and the receiver share $\phi_+^{\widetilde{A_0}\widetilde{A_1}}\otimes\phi_+^{\widetilde{A_0'}\widetilde{A_1'}}\otimes\phi_+^{\widetilde{B_0}B_1}$.
    \item \textit{Pre-processing.} The sender measures with POVM $\{\phi_{xz}^{B_0\tilde{B_0}}\}_{x,z}$. Conditioned on the measurement outcome, the sender applies the unitary $U_{xz}^T\otimes U_{xz}^T$ on the $A_0A_0'$ systems. 
    \item \textit{Transmission.} System $A_0$ is sent through $\mc{N}_1$ and system $A_0'$ is through $\mc{N}_2$.
    \item \textit{Post-processing.} The receiver performs measurement $\{\phi_{xz}^{A_1\widetilde{A_1}}\otimes\phi_{xz}^{A_1'\widetilde{A_1'}},M_\perp^{A_1\widetilde{A_1}A_1'\widetilde{A_1'}}\}$ where $M_\perp = \mathbb{I} - \sum_{x,z}\phi_{xz}\otimes\phi_{xz}$. If $x,z$ is measured, the receiver declares success and applies the correction unitary $U_{xz}^T$. If, however, $\perp$ is measured, the receiver declares failure. 
\end{enumerate}
\end{protocol}
This construction extends the one given in Ref. \cite{ji2024postselected}. In essence, this construction combines the ideas of teleportation, super-dense coding, and error detection. The sender and the receiver attempt to perform a teleportation protocol, where the sender's measurement outcomes $x$ and $z$ will inform the receiver which correction unitary should be applied. However, as they are not allowed free classical communication, the sender sends the measurement outcomes through the given quantum channels by super-dense coding. Since they are given two quantum channels $\mc{N}_1$ and $\mc{N}_2$, they can use the redundancy to perform error detection. That is, $(x,z)$ are encoded on both $\phi_+^{\widetilde{A_0}\widetilde{A_1}}$ and $\phi_+^{\widetilde{A_0'}\widetilde{A_1'}}$, resulting in the state $\phi_{xz}^{\widetilde{A_0}\widetilde{A_1}}\otimes\phi_{xz}^{\widetilde{A_0}\widetilde{A_1}}$. The receiver, upon receiving the two output states, measures each in the generalized Bell basis. If the two measurement results $(\hat{x}_1,\hat{z}_1)$ and $(\hat{x}_2,\hat{z}_2)$ are different (in other words, the measurement result is $\perp$), the receiver knows that an error has occurred and declares failure. If $(\hat{x}_1,\hat{z}_1) = (\hat{x}_2,\hat{z}_2)\eqqcolon (\hat{x},\hat{z})$, the receiver assumes that no error occurred and performs the unitary correction $U_{\hat{x}\hat{z}}^T$ on $B_1$ accordingly. Note that this can also be regarded as a one-shot version of the entanglement-assisted quantum communication protocol introduced in Ref. \cite{devetak2004family}.

Let us now analyze the performance of this protocol. The protocol succeeds when $(\hat{x}_1,\hat{z}_1)=(\hat{x}_2,\hat{z}_2)$, which happens with probability
\begin{align}
    P = \sum_{x,z,\hat{x},\hat{z}} &\frac{1}{d^2} \tr\left\{\left[\mc{N}_1\otimes\mathrm{id}(\phi_{xz})\right]\phi_{\hat{x}\hat{z}}\right\} \notag\\
    &\times \tr\left\{\left[\mc{N}_2\otimes\mathrm{id}(\phi_{xz})\right]\phi_{\hat{x}\hat{z}}\right\}.
\end{align}
Now,
\begin{align}
    &\tr\left\{\left[\mc{N}_i\otimes\mathrm{id}(\phi_{xz})\right]\phi_{\hat{x}\hat{z}}\right\} = \tr\left\{\left[\mc{N}_i\otimes\mathrm{id}(\phi_{00})\right]\phi_{\hat{x}-x,\hat{z}-z}\right\}.
\end{align}
Let $F_{xz}^{(i)}=\tr\left\{\left[\mc{N}_i\otimes\mathrm{id}(\phi_{00})\right]\phi_{xz}\right\}$. The success probability of the protocol is
\begin{align}
    P = \sum_{x,z} F_{xz}^{(1)}F_{xz}^{(2)}.
\end{align}

Suppose the initial state on $B_0$ is $\rho$. The state on $B_1$ before the receiver applies the unitary correction is $U_{xz}^*\rho U_{xz}^T$. With probability $F_{\hat{x}-x,\hat{z}-z}^{(1)}F_{\hat{x}-x,\hat{z}-z}^{(2)}$, the receiver measures $(\hat{x},\hat{z})=(\hat{x}_1,\hat{z}_1)=(\hat{x}_2,\hat{z}_2)$ and the unitary correction $U_{\hat{x}\hat{z}}^T$ is applied. The final state on $B_1$ is then $U_{\hat{x}\hat{z}}^T U_{xz}^* \rho U_{xz}^T U_{\hat{x}\hat{z}}^* = U_{\hat{x}-x,\hat{z}-z}^T\rho U_{\hat{x}-x,\hat{z}-z}^*$. This can be seen using the fact that $Z^zX^x = X^xZ^z$ up to a global phase factor. Therefore, the final channel conditioned on protocol success is given by
\begin{align}
    \N_{out}(\rho) = \sum_{xz} \frac{F_{xz}^{(1)}F_{xz}^{(2)}}{P} U_{xz}^T\rho U_{xz}^*.
\end{align}
It is not hard to see that the worst case fidelity of this channel is attained by $\phi_+$, and therefore the fidelity of this channel is simply 
\begin{align}
    F(\mc{N}_{out}) = F_e(\mc{N}_{out}) = \frac{F_{00}^{(1)}F_{00}^{(2)}}{\sum_{x,z}F_{xz}^{(1)}F_{xz}^{(2)}}.
\end{align}
Now, whenever the input channels are identity-biased, that is, $F_{00}^{(i)}>F_{xz}^{(i)}$ for all $(x,z)\neq(0,0)$, 
\begin{align}
    F(\mc{N}_{out})>\frac{F_{00}^{(1)}F_{00}^{(2)}}{\sum_{x,z}F_{00}^{(1)}F_{xz}^{(2)}} = F_{00}^{(2)},
\end{align}
which is precisely the entanglement fidelity of $\mc{N}_2$ and is greater than or equal to the worst-case fidelity $F(\mc{N}_2)$. By the same reasoning $F(\mc{N}_{out})>F(\mc{N}_1)$. Therefore, the protocol universally improves channel fidelity if the channels are promised to be identity-biased. 

This can be readily generalized to the case where we are given $n>2$ channels by running super-dense coding $n$ times in parallel. The receiver will only accept when all $n$ measurement outcomes coincide. The protocol will succeed with probability 
\begin{align}
    P = \sum_{x,z} \prod_i F_{xz}^{(i)}
\end{align}
and conditioned on success achieves fidelity
\begin{align}
    F(\mc{N}_{out}) = \frac{\prod_i F_{00}^{(i)}}{\sum_{x,z} \prod_i F_{xz}^{(i)}}.
\end{align}

The protocol can be made universal for a wider class of channels by slightly sacrificing the performance. To do so, we first perform channel twirling \cite{horodecki1999general} on each given channel $\mc{N}_i$:
\begin{align}
    \mc{T}[\mc{N}](\rho) &= \int_{U\sim\mathrm{Haar}} U\mc{N}(U^\dagger\rho U)U^\dagger dU. \label{eq:quantum-twirling}
\end{align}
After twirling, each channel becomes a depolarizing channel:
\begin{align}
    \mc{T}[\mc{N}_i](\rho) = \mc{D}_{F_i}(\rho) \coloneqq F_i\rho + (1-F_i)\frac{d\mathbb{I}_d-\rho}{d^2-1}
\end{align}
where $F_i=F_e(\mc{N}_i)$. Then, we perform our protocol on the twirled channels $\mc{T}[\mc{N}_i]$, achieving fidelity
\begin{align}
    F(\mc{N}_{out}) = \frac{\prod_i F_i}{\prod_i F_i + \left(\frac{1}{d^2-1}\right)^{n-1}\prod_i(1-F_i)} \label{eq:output-fidelity}
\end{align}
with probability
\begin{align}
    P = \prod_i F_i + \left(\frac{1}{d^2-1}\right)^{n-1}\prod_i(1-F_i).
\end{align}
The achieved fidelity $F(\mc{N}_{out})>F_i$ for all $i$ when $1/d^2<F_i<1$. To see this, note that $F(\mc{N}_{out})$ monotonically increases with $F_i$. So, 
\begin{align}
    &F(\N_{out}) \notag\\
    &> \frac{\frac{F_i}{(d^2)^{n-1}}}{\frac{F_i}{(d^2)^{n-1}} + \left(\frac{1}{d^2-1}\right)^{n-1}(1-F_i)\left(1-\frac{1}{d^2}\right)^{n-1}} \\
    &= F_i =F_e(\mc{N}_i) \geq F(\N_i)
\end{align}

Let us now illustrate the performance with an example. Suppose we are given two qubit channels with $F_1\coloneqq F_e(\mc{N}_1)=0.9$ and $F_2\coloneqq F_e(\mc{N}_2)=0.8$. The protocol succeeds with probability
\begin{align}
    P = F_1F_2+\frac{1}{d^2-1}(1-F_1)(1-F_2) \approx 0.727,
\end{align}
while the fidelity of the output channel is
\begin{align}
    F = \frac{F_1F_2}{F_1F_2+\frac{1}{d^2-1}(1-F_1)(1-F_2)} \approx 0.991.
\end{align}
That is, with approximately 73\% probability we get a channel that is almost noiseless. The protocol performs even better when $d$ is large, achieving fidelity $1-O(1/d^{2n})$ with probability $F_1F_2+O(1/d^{2n})$.

Furthermore, recall that a depolarizing channel $\mc{D}$ is entanglement-breaking if $F_e(\mc{D})\leq1/d$. Thus, an entanglement-breaking channel, which cannot transmit any quantum information, can be used to elevate the ability of an entanglement-preserving channel to transmit quantum information. Perhaps more strikingly, even if neither channel can transmit any quantum information, an output channel that is entanglement-preserving may be extracted by using our encoding and decoding strategy. For instance, suppose we are given two qubit depolarizing channels with $F_1=0.4$ and $F_2=0.35$, our entanglement-assisted strategy succeeds with probability 0.27, and the entanglement fidelity of the output channel is $14/27\approx 0.519 > 0.5$.

\subsection{Optimality}
It turns out that our protocol achieves the optimal fidelity when the given channels are all depolarizing channels. That is, even if we know precisely the depolarizing parameters of these channels and are allowed to optimize over all possible NS superinstruments, the achievable fidelity is still upper bounded by the right hand side of Eq.~\eqref{eq:output-fidelity}. More generally, we have the following fidelity upper bound when the given channels are generalized Pauli channels.
\begin{proposition}\label{prop:optimality}
    Given $n$ generalized Pauli channels $\{\mc{N}_i\}_{i=1}^n$ with $\mc{N}_i(\rho) = \sum_{x,z} p_{xz}^{(i)} U_{xz}\rho U_{xz}^\dagger$, let $p_{max}^{(i)}=\max_{x,z}\{p_{xz}^{(i)}\}$ and $p_{min}^{(i)}=\min_{x,z}\{p_{xz}^{(i)}\}$. For any NS-assisted superinstrument $\{\Theta_0,\Theta_1\}$ such that the output channel can be defined as $\mc{N}_{out}=\displaystyle\frac{\Theta_0[\mc{N}_1\otimes\cdots\otimes\mc{N}_n]}{p(\Theta_0;\mc{N}_1,\cdots,\mc{N}_n)}$, we have
    \begin{align}
        F\left(\N_{out}\right) \leq F_e(\mc{N}_{out}) \leq \frac{\prod_i p_{max}^{(i)}} {\prod_i p_{max}^{(i)}+(d^2-1)\prod_i p_{min}^{(i)}}.
    \end{align}
\end{proposition}
We will prove this proposition in the Supplementary Material. From this expression, we immediately conclude that our protocol achieves the optimal fidelity when the given channels are depolarizing channels $\mc{D}_{F_i}$ with $F_i\geq1/d^2$. This also shows that it is not possible to increase the fidelity of a single depolarizing channel by NS superinstruments. Moreover, our protocol achieves this optimal fidelity with the largest possible probability.
\begin{proposition}\label{prop:optimal-prob}
    The probability of transforming depolarizing channels $\mc{D}_{F_1}\otimes\mc{D}_{F_2}$ into an output channel with fidelity 
    \begin{align}
        F = \frac{F_1F_2}{F_1F_2+\frac{1}{d^2-1}(1-F_1)(1-F_2)}
    \end{align}
    using NS superinstruments is upper bounded by 
    \begin{align}
        P \leq F_1F_2 + \frac{1}{d^2-1}(1-F_1)(1-F_2).
    \end{align}
\end{proposition}

Our construction is also essentially the unique NS superinstrument that is universal for channels with entanglement fidelity greater than $1/d^2$. All other universal superinstruments reduce to our construction upon twirling the channels and rescaling the success probability.
\begin{proposition}\label{prop:uniqueness}
    Let the superinstrument that describes our construction be $\Theta^\mathrm{EA}=\{\Theta_0^{\mathrm{EA}},\Theta_1^{\mathrm{EA}}\}$. Any NS-assisted superinstrument $\{\Theta_0,\Theta_1\}$ that is universal for arbitrary channels with fidelity greater than $1/d^2$ satisfies
    \begin{align}
        C\Theta_0^{\mathrm{EA}}[\cdot] = \mc{T}^B\circ\Theta_0^{AA'\to B}\left[\mc{T}^A\otimes\mc{T}^{A'}[\cdot]\right].
    \end{align}
    for some constant $C\in(0,1]$ that reduces the success probability.
\end{proposition}
This result is a special case of Thm. \ref{thm:uniqueness}, which will be presented later and proven in the Supplementary Material.

\subsection{Classical channels}
The same idea applies to classical channels as well. A classical channel $\mc{P}$ given by transition probabilities $\mc{P}(y|x)$ can be written in bra-ket notation as
\begin{align}
    \mc{P}(W) = \sum_{x,y} \mc{P}(y|x)\op{y}{x}W\op{x}{y}.
\end{align}  
We can analogously define its fidelity as the worst-case correct transmission probability:
\begin{align}
    F(\mc{P}) = \min_{x} \mc{P}(x|x).
\end{align}
The following shared-randomness-assisted protocol universally improves the fidelity of classical channels.

\begin{protocol}
The shared-randomness protocol proceeds as follows (see Fig. \ref{fig:SR-protocol} for an illustration).
\begin{figure*}
    \centering
    \includegraphics[width=0.8\linewidth]{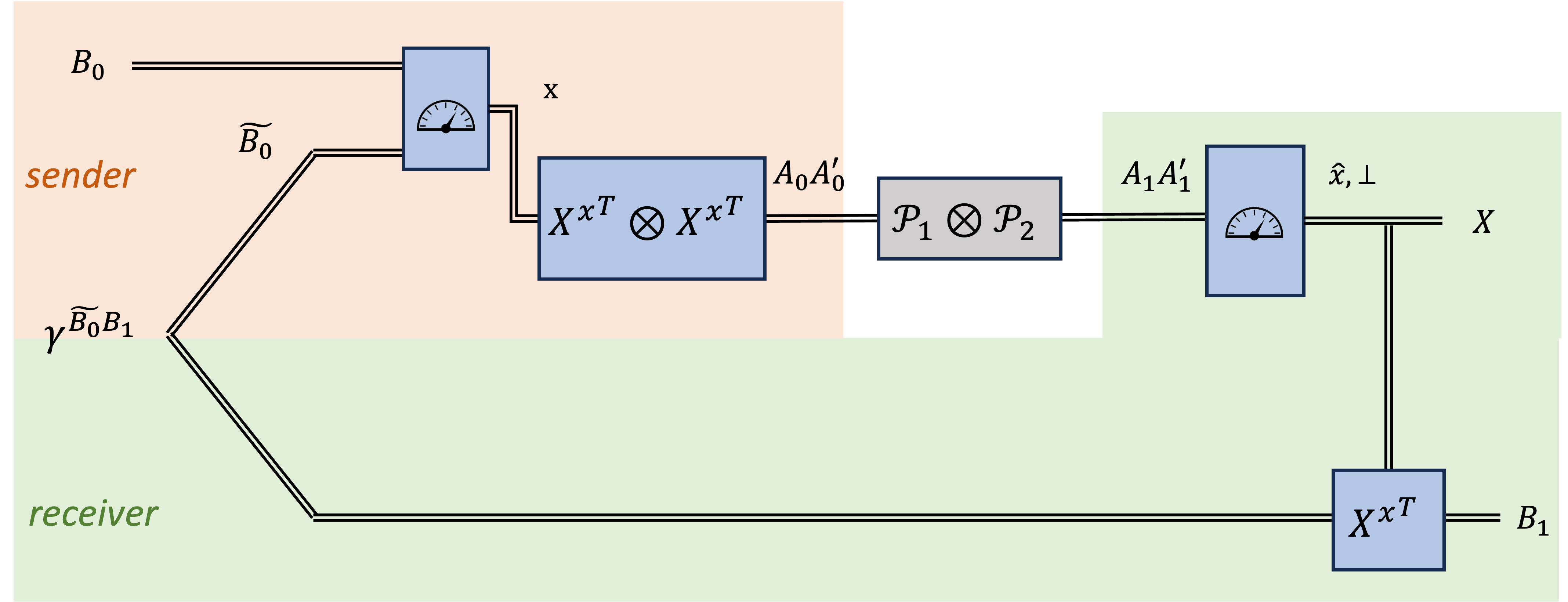}
    \caption{Implementation of the shared-randomness protocol.}
    \label{fig:SR-protocol}
\end{figure*}
\begin{enumerate}
    \item \textit{Shared resource.} The sender and the receiver share $\gamma^{\widetilde{B_0}B_1}$, where $\gamma$ is the classically correlated state $\gamma=\sum_i\frac{1}{d}\state{ii}$.
    \item \textit{Pre-processing.} The sender measures with POVM $\{\Gamma_x^{B_0\tilde{B_0}}\}_x$ where $\Gamma_x=(\mbb{I}\otimes X^x)(\sum_i \state{ii})(\mbb{I}\otimes X^x)^\dagger$. The sender applies $(X^x\otimes X^x)^T$ conditioned on the measurement outcome. 
    \item \textit{Transmission.} System $A_0$ is sent through $\mc{P}_1$ and system $A_0'$ is through $\mc{P}_2$.
    \item \textit{Post-processing.} The receiver measures $\{\state{xx},M_\perp\}$ where $M_\perp=\mbb{I}-\sum_x\state{xx}$. If $x$ is measured, the receiver declares success and applies the correction unitary $(X^x)^T$. If $\perp$ is measured, the receiver declares failure.
\end{enumerate}
\end{protocol}
The idea is completely analogous to the entanglement-assisted protocol. Suppose the input at $B_0$ is some classical state $W=\sum_i w_i\state{i}$. After the sender performs the measurement $\{\Gamma_x\}_x$, the output state on $B_1$ is $X^x W (X^x)^\dagger$. The measurement outcome is communicated to the receiver by using the two channels $\mc{P}_1\otimes\mc{P}_2$. The probability that the receiver measures $\hat{x}\coloneqq\hat{x}_1=\hat{x}_2$ is $\sum_x \frac{1}{d}\mc{P}_1(\hat{x}|x)\mc{P}_2(\hat{x}|x)$, and so the probability that the receiver declares success is
\begin{align}
    P = \sum_{x,\hat{x}} \frac{1}{d}\mc{P}_1(\hat{x}|x)\mc{P}_2(\hat{x}|x).
\end{align}
Upon measuring $\hat{x}$, the receiver applies the correction unitary $(X^{\hat{x}})^T$. So, the final state on $B_1$ is $(X^{\hat{x}})^{T} X^{x}W(X^x)^\dagger(X^{\hat{x}})^* = X^{x-\hat{x}}W(X^{x-\hat{x}})^\dagger$. The resulting channel is therefore
\begin{align}
    \mc{P}_{out}(W) = \sum_{x,\hat{x}} \frac{\mc{P}_1(\hat{x}|x)\mc{P}_2(\hat{x}|x)}{dP} X^{x-\hat{x}}W(X^{x-\hat{x}})^\dagger
\end{align}
The worst-case fidelity of $\mc{P}_{out}$ is precisely
\begin{align}
    F(\mc{P}_{out}) = \frac{\sum_x\mc{P}_1(x|x)\mc{P}_2(x|x)}{\sum_{x,\hat{x}}\mc{P}_1(\hat{x}|x)\mc{P}_2(\hat{x}|x)}.
\end{align}
If we are given $n>2$ channels, the fidelity of $\mc{P}_{out}$ is
\begin{align}
    F(\mc{P}_{out}) = \frac{\sum_x \prod_i\mc{P}_i(x|x)}{\sum_{x,\hat{x}}\prod_i \mc{P}_i(\hat{x}|x)}.
\end{align}
Analogous to the entanglement-assisted protocol, we conclude that the protocol is universal for all identity-biased classical channels, which are channels such that $\mc{P}(x|x) > \mc{P}(\hat{x}|x)$ for all $\hat{x}\neq x$. This is because
\begin{align}
    F(\mc{P}_{out}) &> \frac{\sum_x \mc{P}_i(x_{min}|x_{min})\prod_{j\neq i}\mc{P}_j(x|x)}{\sum_{x,\hat{x}} \mc{P}_i(\hat{x}|x)\prod_{j\neq i}\mc{P}_j(x|x)} \\
    &= \frac{\sum_x F(\mc{P}_i)\prod_{j\neq i}\mc{P}_j(x|x)}{\sum_{x} \prod_{j\neq i}\mc{P}_j(x|x)} \\
    &= F(\mc{P}_i),
\end{align}
where $x_{min}$ is $\mathrm{argmin}_x \mc{P}_i(x|x)$. Again, we can enlarge the set over which the protocol is universal by twirling the input channels classically:
\begin{align}
    \mc{T}_{cl}[\mc{P}](W) \coloneqq \frac{1}{d!}\sum_{\sigma\in\text{Sym}(d)} U_\sigma^\dagger\mc{P}(U_\sigma W U_\sigma^\dagger)U_\sigma \label{eq:classical-twirling}
\end{align}
where $U_\sigma$ is the unitary representation of permutation $\sigma$. The twirled channel becomes a $d$-ary symmetric channel (dSC)
\begin{align}
    \mc{T}[\mc{P}_i](W) = \mathrm{dSC}_{F_i}(W) \coloneqq F_iW + (1-F_i)\frac{\mathbb{I}_d-W}{d-1}.
\end{align}
where $F_i=F_{avg}(\mc{P}_i)\coloneqq \sum_{x}\frac{1}{d}\mc{P}_i(x|x)$. Performing our protocol on the twirled channels yields an output channel with fidelity
\begin{align}
    F(\mc{P}_{out}) = \frac{\prod_i F_i}{\prod_i F_i + \left(\frac{1}{d-1}\right)^{n-1}\prod_i (1-F_i)}.
\end{align}
One can readily show that $F(\mc{P}_{out})>F_i$ when $F_i>1/d$. 
The fidelity upper bound in the previous section follows analogously.
\begin{proposition}\label{prop:classical-optimality}
    Given $n$ classical channels $\{\mc{P}_i\}_{i=1}^n$ with $\mc{P}_i(W) = \sum_{x,y} \mc{P}_i(y|x)\op{y}{x}W\op{x}{y} $, let $\Omega_i = \max_y \displaystyle\frac{\max_x\mc{P}_i(y|x)}{\min_x\mc{P}_i(y|x)}$. For any NS-assisted superinstrument $\{\Theta_0,\Theta_1\}$ such that the output channel can be defined as $\mc{P}_{out}=\displaystyle\frac{\Theta_0[\mc{P}_1\otimes\cdots\otimes\mc{P}_n]}{p(\Theta_0;\mc{P}_1,\cdots,\mc{P}_n)}$, we have
    \begin{align}
        F\left(\mc{P}_{out}\right) \leq \frac{\prod_i \Omega_i} {\prod_i \Omega_i+(d-1)}.
    \end{align}
\end{proposition}
For symmetric channels $\mathrm{dSC}_{F_i}$, $\Omega_i = \frac{F_i}{1-F_i}(d-1)$. So, as a corollary of this proposition, our protocol is optimal when the $n$ channels are all symmetric channels. We also show that our shared-randomness protocol is essentially the unique universal protocol for classical channels, as formally stated below. This is again a special case of Thm. \ref{thm:uniqueness}.
\begin{proposition}\label{prop:classical-uniqueness}
    Let the superinstrument that describes our shared-randomness protocol be $\Theta^\mathrm{SR}=\{\Theta_0^{\mathrm{SR}},\Theta_1^{\mathrm{SR}}\}$. Any NS-assisted superinstrument $\{\Theta_0,\Theta_1\}$ that is universal for arbitrary classical channels with fidelity greater than $1/d$ satisfies
    \begin{align}
        C\Theta_0^{\mathrm{SR}}[\cdot] = \mc{T}_{cl}^B\circ\Theta_0^{AA'\to B}\left[\mc{T}_{cl}^A\otimes\mc{T}_{cl}^{A'}[\cdot]\right].
    \end{align}
    for some constant $C\in(0,1]$ that reduces the success probability.
\end{proposition}

\subsection{Generic resource theories}
Finally, we will prove general statements on the possibility of universal distillation of a resource. For full generality, we will state and prove our result for dynamical resource theories, which will reduce to static resource theories (i.e., resource theories of quantum states) by fixing the input system to be a trivial one-dimensional system. In the following, we will always assume that the resource theory is equipped with a free generalized channel twirling, defined as follows.
\begin{definition}[\cite{takagi2022one}]
    A generalized channel twirling is a superchannel defined by an operator $0\leq P\leq\mathbb{I}$ as
    \begin{align}
        \mc{T}[\mc{N}]=\tr(J_\mc{N} P)\Phi+\tr(J_\mc{N}(1-P))\Lambda
    \end{align}
    where $\Phi$ is a resource channel, $\tr(J_\Phi P)=1$, and $\tr(J_\Lambda J_\Phi)=\tr(J_\Lambda P)=0$.
\end{definition}
Generalized channel twirlings are indeed superchannels because they admit a physical decomposition shown in Fig. \ref{fig:twirling-superchannel}. In fact, they are a type of measure-and-prepare superchannels \cite{chen2020entanglement}. In static resource theories, they reduce to generalized twirling channels. Most standard resource theories are equipped with such a free generalized twirling. This includes, but is not limited to, the static resource theories of entanglement, coherence \cite{aberg2006quantifying,streltsov2017colloquium}, thermodynamics \cite{brandao2013resource}, and asymmetry \cite{bartlett2007reference}, as well as the dynamical resource theories of entanglement \cite{gour2020dynamicalentanglement}, coherence \cite{saxena2020dynamical}, and communication. The resource channel $\Phi$ is precisely the target channel that we will try to distill, and a channel of the form 
\begin{align}
    \N_F = F\Phi+(1-F)\Lambda
\end{align}
will be referred to as an isotropic channel.

\begin{figure}
    \centering
    \includegraphics[width=\linewidth]{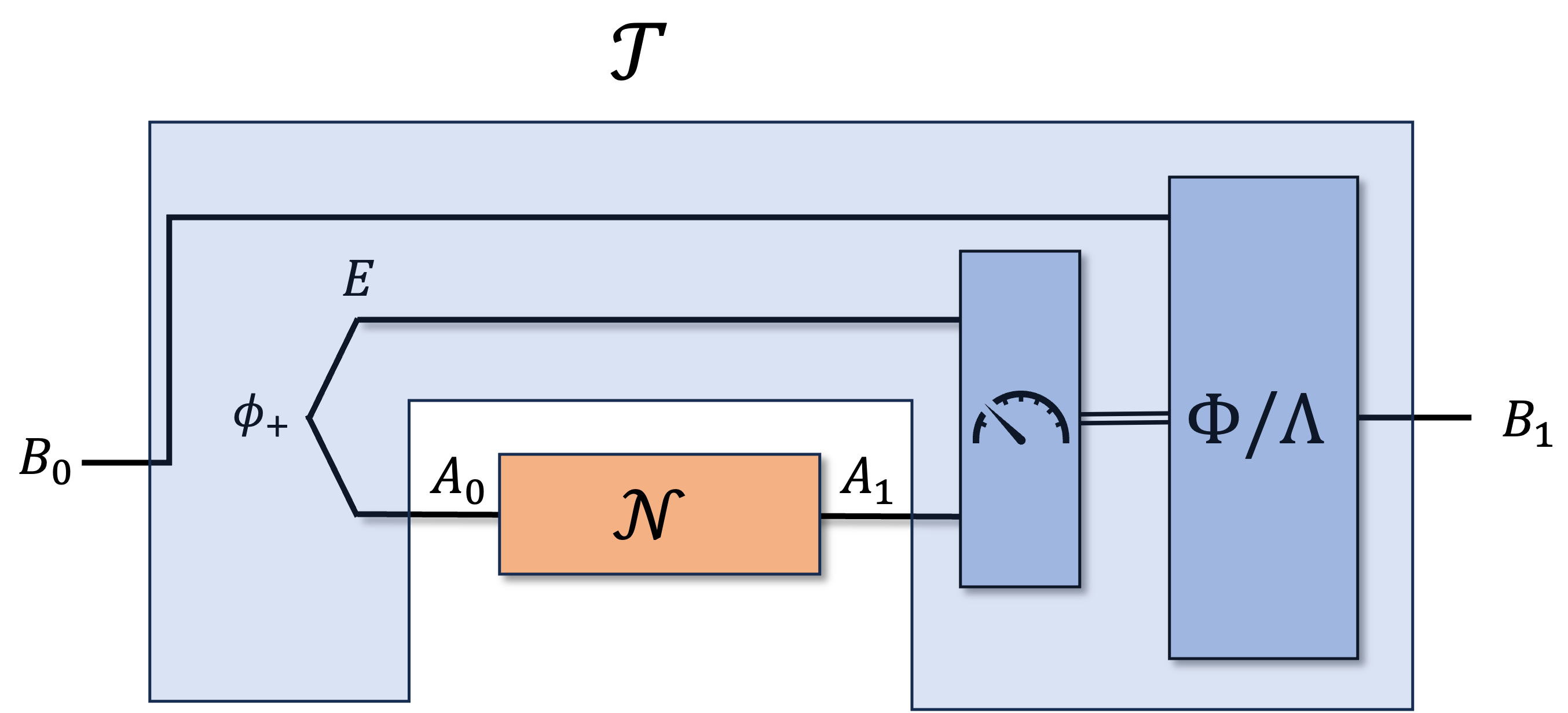}
    \caption{A physical decomposition of the twirling superchannel. The state $\phi_+^{A_0E}$ is prepared and sent through the channel, and a POVM $\{P,\mbb{I}-P\}$ is performed on $A_1E$. The quantum state on $B_0$ is sent through $\Phi$ or $\Lambda$ depending on the measurement outcomes.}
    \label{fig:twirling-superchannel}
\end{figure}

To state the theorem, a few more definitions are in order. Given a resource theory with free channels $\mbb{F}_C$, the standard robustness and generalized robustness are \cite{vidal1999robustness,datta2009min} 
\begin{align}
    &R_s(\N) = \min\left\{1+s:\frac{\N+s\M}{1+s}\in\mbb{F}_C,\ \M\in\mbb{F}_C\right\}, \\
    &R_g(\N) = \min\left\{1+s:\frac{\N+s\M}{1+s}\in\mbb{F}_C,\ \M\in\text{CPTP}\right\}.
\end{align}
In words, the standard robustness quantifies the amount of free disturbance required to make $\N$ free, while the generalized robustness quantifies the amount of arbitrary disturbance that is required. These can be rephrased as
\begin{align}
    &R_s(\N) = \min\left\{\lambda: \N\leq_{\mbb{F}_C}\lambda\M,\ \M\in\mbb{F}_C\right\}, \\
    &R_g(\N) = \min\left\{\lambda:\N\leq_{c.p.}\lambda\M,\ \M\in\mbb{F}_C\right\},
\end{align}
where $\N\leq_{\mbb{F}_C}\N'$ if $\N'-\N\in\text{cone}(\mbb{F}_C)=\{\lambda\M:\M\in\mbb{F}_C\}$ and $\N\leq_{c.p.}\N'$ if $\N'-\N$ is completely positive. The projective robustness of a resource is defined as \cite{regula2022probabilistic,regula2022tight} 
\begin{align}\label{eq:proj-rob}
    \Omega(\mc{N})=\inf\{\gamma:\mc{N}\leq_{c.p.}\lambda\mc{M}\leq_{c.p.}\gamma\mc{N},\ \mc{M}\in\mathbb{F}_C\}.
\end{align}
Finally, following Ref. \cite{gour2019comparison}, we will define the inner product between two completely positive maps by $\langle\mc{N},\mc{M}\rangle = \tr(J_\mc{N}J_{\mc{M}})$, where $J_\N$ is the normalized Choi state of $\N$ defined as $\mathrm{id}\otimes\N(\phi_+)$. Let $F_c \coloneqq \max_{\M\in\mbb{F}_C}\langle\Phi,\M\rangle$. The following theorem provides a sufficient condition for the existence of a resource-non-generating superinstrument for universal resource distillation.
\begin{theorem}\label{thm:general-rt-suff-condition}
    In any resource theory with a free generalized channel twirling $\mc{T}$: $\mc{T}[\mc{N}]=\tr(PJ_\mc{N})\Phi+\tr[(\mathbb{I}-P)J_\mc{N}]\Lambda$, if 
    \begin{enumerate}
        \item for some isotropic channel $\mc{N}_F=F\Phi+(1-F)\Lambda$ with $F\in(F_c,1)$, the projective robustness is multiplicative, i.e., $\Omega(\mc{N}_F^{\otimes 2})=\Omega(\mc{N}_F)^2$, and
        \item the resource theory is either
        \begin{enumerate}
            \item affine, or
            \item convex with $R_s(\Phi)=R_g(\Phi)$,
        \end{enumerate}
    \end{enumerate}
    then the supermap 
    \begin{align}\label{eq:general-supermap}
        \Theta_0[\mc{N}] =& \tr[P\otimes P J_\mc{N}]\Phi \notag\\
        &+ \frac{F_c}{1-F_c}\tr[(\mbb{I}-P)\otimes(\mbb{I}-P)J_\mc{N}]\Lambda
    \end{align}
    is resource-non-generating (RNG), i.e., $\Theta_0[\M]$ is proportional to a free channel for every $\M\in\mbb{F}_C$. Moreover, $\Theta_0[\N_1\otimes\N_2]=p\N_{out}$ with $\langle\N_{out},\Phi\rangle>\max \{\langle\N_1,\Phi\rangle,\langle\N_2,\Phi\rangle\}$ whenever $\langle\mc{N}_i,\Phi\rangle>F_c$. 
\end{theorem}

The proof of Thm. \ref{thm:general-rt-suff-condition} will be provided in the Supplementary Material. Instead, we offer a few remarks on the applicability of the theorem. We believe that multiplicativity of projective robustness needs to be checked on a case-by-case basis for each resource theory. However, the following result shows that resource theories with a certain structure must have multiplicative projective robustness for isotropic channels. While not all resource theories satisfy the following condition, we do not know of a resource theory in which multiplicativity fails to hold for isotropic channels. 
\begin{proposition}
    Let $\mbb{F}_C(A)$ denote the set of free channels on dynamical system $A$, and let $A'$ be a dynamical system isomorphic to $A$. If $\mbb{F}_C(A\otimes A')$ is the convex hull or affine hull of $\mbb{F}_C(A)\otimes\mbb{F}_C(A')$, then $\Omega(\N_F^{\otimes2}) = \Omega(\N_F)^2$ for any $F>F_c$.
\end{proposition}
\begin{proof}
    Let us first assume $\mbb{F}_C(A\otimes A')$ is the convex hull of $\mbb{F}_C(A)\otimes\mbb{F}_C(A')$. Since $\N_F$ is invariant under generalized twirling, for any feasible solution $\mc{M}_1^A$ of $\Omega(\mc{N}_F)$, $\mc{T}[\mc{M}_1]$ is also a feasible solution which attains the same $\gamma$. Therefore, we can without loss of generality choose $\mc{M}_1$ to be a free isotropic channel. Then, we can readily see that the optimal solution is $\N_{F_c}$, and $\Omega(\N_F)=\frac{F}{1-F}\frac{1-F_c}{F_c}$. Now, since $\mc{N}_F^{A}\otimes\mc{N}_F^{A'}$ is invariant under $\mc{T}\otimes\mc{T}$, we can also assume that the optimal solution $\mc{M}_2^{AA'}$ for $\Omega(\N_F^{\otimes 2})$ is invariant under $\mc{T}\otimes\mc{T}$. Now, writing such a free $\mc{M}_2$ as a convex hull of tensor products of free channels $\mc{M}_2^{AA'}=\sum_i p_i\mc{A}_i^A\otimes\mc{B}_i^{A'}$ where $\mc{A}_i$ and $\mc{B}_i$ are free, we obtain
    \begin{align}
        &\mc{T}\otimes\mc{T}[\mc{M}_2^{AA'}] \notag\\
        &= \sum_i p_i\mc{T}[\mc{A}_i]\otimes\mc{T}[\mc{B}_i] \\
        &= \sum_i p_i \mc{N}_{F_i}\otimes\mc{N}_{F'_i} \\
        &= a_1\Phi\otimes\Phi + a_2\Phi\otimes\Lambda + a_3\Lambda\otimes\Phi + a_4\Lambda\otimes\Lambda
    \end{align}
    where $a_1=\sum_i p_i F_iF'_i$, $a_2=\sum_i p_iF_i(1-F_i')$, $a_3=\sum_i p_i(1-F_i)F'_i$, $a_4=\sum_i p_i(1-F_i)(1-F'_i)$ and $F_i,F_i'\leq F_c\ \forall i$. The smallest $\lambda$ such that $\N_F\otimes\N_F \leq_{c.p.} \lambda\M_2$ is 
    \begin{align}
        \lambda^* = \min\max\left\{\frac{F^2}{a_1}, \frac{F(1-F)}{a_2}, \frac{F(1-F)}{a_3}, \frac{(1-F)^2}{a_4}\right\}
    \end{align}
    It is not hard to see that 
    \begin{align}
        \frac{a_1}{a_2} &= \frac{\sum_i p_i F_iF'_i}{\sum_i p_i F_i(1-F'_i)} \leq \frac{\sum_i p_iF_iF_c}{\sum_i p_iF_i(1-F_c)} \\
        &= \frac{F_c}{1-F_c} < \frac{F^2}{F(1-F)}.
    \end{align}
    Therefore, $F^2/a_1 > F(1-F)/a_2$. Similarly $F^2/a_1 > F(1-F)/a_3$ and $F^2/a_1 > (1-F)^2/a_4$. This shows that $\lambda=\min\{F^2/a_1\} \geq F^2/F_c^2$. Likewise, we can show that the smallest $\mu$ such that $\M_2\leq\mu\N_F\otimes\N_F$ is $\mu^*=\min\{a_4/(1-F)^2\}\geq (1-F_c)^2/(1-F)^2$. Therefore, $\Omega(\N_F\otimes\N_F)=\frac{F^2(1-F_c)^2}{(1-F)^2F_c^2}$. This proves the multiplicativity of projective robustness. 

    The proof for when $\mbb{F}_C(A\otimes A')$ is the affine hull of $\mbb{F}_C(A)\otimes\mbb{F}_C(A')$ proceeds analogously. First, note that since $\mbb{F}_C(A\otimes A')$ is an affine set, both $\mbb{F}_C(A)$ and $\mbb{F_C}(A')$ must also be affine sets. This is because $\mbb{F}_C(A)$ can be obtained by tracing out the $A'$ system of channels in $\mbb{F_C}(A\otimes A')$. The same is true for $\mbb{F}_C(A')$. Consider a free isotropic channel in $\mbb{F}_C(A)$, which exists because we can choose an arbitrary free channel in $\mbb{F}_C(A)$ and perform the free generalized twirling to obtain a free isotropic channel. However, there can only be one such free isotropic channel: Suppose there exist two free isotropic channels $F_1\Phi+(1-F_1)\Lambda$ and $F_2\Phi+(1-F_2)\Lambda$. By affinity, $\Phi$ must also be in $\mbb{F_C}(A)$, which contradicts our assumption that $\Phi$ is a resource channel.
    
    With this in mind, the optimal solution for $\Omega(\N_F)$ can be chosen to be the unique free isotropic channel $\N_{F_c}$ without loss of generality. let us again write the optimal solution $\M_2$ for $\Omega(\N_F^{\otimes 2})$ as an affine hull $\M_2^{AA'}=\sum_i a_i\mc{A}_i^A\otimes\mc{B}_i^{A'}$ where $\mc{A}_i$ and $\mc{B}_i$ are free, we have, by the uniqueness of the free isotropic channel
    \begin{align}
        \mc{T}\otimes\mc{T}[\mc{M}_2^{AA'}] &= \sum_i a_i\mc{T}[\mc{A}_i]\otimes\mc{T}[\mc{B}_i] \\
        &= \sum_i a_i \N_{F_c}\otimes\N_{F_c} = \N_{F_c}\otimes\N_{F_c}.
    \end{align}
    In other words, we can without loss of generality choose the optimal solution for $\Omega(\N_F^{\otimes 2})$ to be $\N_{F_c}\otimes\N_{F_c}$. This proves multiplicativity.
\end{proof}

The condition $R_s(\Phi)=R_g(\Phi)$ is satisfied by the maximally resourceful state in a wide range of convex resource theories, including quantum entanglement, genuine multipartite entanglement \cite{contreras2019resource}, and multi-level coherence \cite{johnston2018evaluating}. In the resource theory of non-stabilizerness, there also exist states where this condition holds, for instance the three-qubit Hoggar state, the qutrit Strange state, and the qutrit Norrell state \cite{takagi2022one}. 

In the case of resource theory of communication, the free generalized twirling is precisely the channel twirling operations in Eq.~\eqref{eq:quantum-twirling} and Eq.~\eqref{eq:classical-twirling}. The target resource is the quantum or classical identity channel, with $\langle\mathrm{id},\N\rangle=\tr[\phi_+(\mathrm{id}\otimes\N)(\phi_+)]=F_e(\N)$. Furthermore, it is not hard to see that the resource theory is affine, and that $\mbb{F}_C(A\otimes A')$ is an affine hull of $\mbb{F_C}(A)\otimes\mbb{F}_C(A')$. The conditions in Thm. \ref{thm:general-rt-suff-condition} are thus satisfied, and $F_c=\max_{\M = \tr(\cdot)\omega}\langle\mathrm{id},\M\rangle = \max_{\omega}\tr[\phi_+(\mbb{I}/d\otimes\omega)] = 1/d^2$. Thus, Thm. \ref{thm:general-rt-suff-condition} implies the existence of an NS superinstrument that is universal for channels with entanglement fidelity greater than $1/d^2$. It does not, however, tell us that the superinstrument turns out to be implementable physically by entanglement assistance as well.

As we have alluded to in previous sections, the superinstrument constructed in Thm. \ref{thm:general-rt-suff-condition} is essentially unique.
\begin{theorem}\label{thm:uniqueness}
    The supermap in Eq.~\eqref{eq:general-supermap} is unique up to twirling the input channels and the output channel and rescaling the success probability. More concretely, any other $\Theta_0'$ such that $\Theta_0'[\N_1\otimes\N_2]=p\N_{out}$ with $\langle\N_{out},\Phi\rangle>\max \{\langle\N_1,\Phi\rangle,\langle\N_2,\Phi\rangle\}$ whenever $\langle\mc{N}_i,\Phi\rangle>F_c$ must satisfy
    \begin{align}
        C\Theta_0 = \mc{T}\circ\Theta'_0\circ\mc{T}^{\otimes 2}
    \end{align}
    for some positive number $C\in(0,1]$.
\end{theorem}

By leveraging the uniqueness theorem, we can show that the sufficient conditions stated in Thm. \ref{thm:general-rt-suff-condition} turn out to be almost necessary.
\begin{theorem}\label{thm:general-rt-nec-condition}
    In a convex resource theory with a generalized twirling, if 
    \begin{enumerate}
        \item the projective robustness is strictly sub-multiplicative for some isotropic channel, or
        \item $R_g(\Phi)<R_s(\Phi)<\infty$,
    \end{enumerate}
    then $\Theta_0$ in Eq.~\eqref{eq:general-supermap} is not RNG, implying that universal distillation is not possible.
\end{theorem}

The proofs of these two theorems can be found in the Supplementary Material.

\section{Discussion}
We studied the problem of universal probabilistic distillation of quantum resources. Despite recent progress that strongly constrained the possibility of universal entanglement distillation and universal state purification, we show that universal improvement of channel fidelities is possible within the context of resource theory of communication. More specifically, we demonstrated an entanglement-assisted universal probabilistic coding strategy for non-i.i.d. channels. Conditioned on success, our protocol yields an output channel with higher worst-case fidelity than all given channels, which can be arbitrarily given as long as they satisfy a weak condition that their fidelities are greater than $1/d^2$. The universality of our protocol makes it very desirable for practical implementations: We are free to blindly apply the protocol to any set of given channels which we do not have knowledge of, so long as their fidelities are not prohibitively low. Our protocol is shown to be optimal for depolarizing channels and essentially the unique protocol that is universal. We then adapt our protocol for fully classical channels. Similar to the protocol for quantum channels, our shared-randomness-assisted protocol yields an output classical channel whose worst-case correct transmission probability is higher than all given classical channels, as long as their worst-case correct transmission probability is greater than $1/d$. Finally, we provide general theorems illuminating the conditions under which a quantum resource can be universally distilled by a resource-non-generating operation.

\bibliographystyle{unsrt}
\bibliography{refs}

\section*{Acknowledgement}
The authors thank Bartosz Regula for the valuable initial discussions at QIP2025 that led us to the resource theory of communication. The authors also thank Allen Zang, Ian George, and Paul G. Kwiat for useful discussions along the way. ChatGPT-5.5 Instant was used in the proof of Proposition \ref{prop:optimal-prob} to solve a feasibility problem. The authors assume full responsibility for the final manuscript. This work is supported by the U.S. Department of Energy Office of Science National Quantum Information Science Research Centers. Results in this paper are a part of X.C.'s Ph.D. dissertation.

\clearpage
\onecolumngrid

\begin{center}
    \mbox{\Large \textbf{Supplementary Material}}
\end{center}

\section{The Superchannel Formalism and Dynamical Resource Theories}
In this section, we review concepts of supermaps that will be necessary for our discussion. For more details, we refer readers to Refs. \cite{gour2019comparison,burniston2020necessary}.

Let $\mc{B}(\mc{H})$ denote the space of all bounded operators acting on a finite dimensional Hilbert space. The space of density operators $\mc{D}(\mc{H})$ is a subspace of $\mc{B}(\mc{H})$ that is positive semi-definite and trace one. The set of linear maps from $\mc{B}(\mc{H}^{A_0})$ to $\mc{B}(\mc{H}^{A_1})$ is denoted as $\mc{L}^{A_0\to A_1}$ or simply as $\mc{L}^A$, and such linear maps are often also referred to as superoperators. The set of quantum channels from $A_0$ to $A_1$ consists of superoperators that are completely positive (CP) and trace-preserving (TP) and forms a subspace of $\mc{L}^A$. For every superoperator $\N^A$, we can define the corresponding (unnormalized) Choi operator 
\begin{align}
    \widetilde{J}_\N^A = \mathrm{id}^{A_0}\otimes\N^{A_0'\to A_1}\left(\sum_{ij}\op{ii}{jj}\right) = \sum_{ij}\op{i}{j}\otimes\mc{N}(\op{i}{j}).
\end{align}
Turns out, this defines an isomorphism between $\mc{L}^A$ and $\mc{B}(\mc{H}^{A})=\mc{B}(\mc{H}^{A_0A_1})$, where the inverse map is given by 
\begin{align}
    \mc{N}^A(\rho) = \tr_{A_0}\left[\widetilde{J}_{\mc{N}}^{A}\left((\rho^T)^{A_0}\otimes\mathbb{I}^{A_1}\right)\right].
\end{align}
This is referred to as the Choi-Jami{\l}kowski isomorphism. We emphasize that, in what follows, ``Choi operator" refers to the unnormalized version, denoted by $\widetilde{J}_\N$, while $J_\N$ denotes specifically the normalized Choi state.

In this language, supermaps are linear maps from superoperators $\mc{L}^A$ to superoperators $\mc{L}^B$. We can analogously define a Choi operator of a supermap as follows. Recall that the canonical orthonormal basis for $\mc{L}^A$ are maps of the form
\begin{align}
    \mc{E}_{ij,kl}^{A}(\rho) = \bra{i}\rho\ket{j}\ket{k}\bra{l}.
\end{align}
The (unnormalized) Choi operator of $\Theta$ is then defined as
\begin{align}
    \wt{J}_\Theta^{AB} = \sum_{ij,kl} \wt{J}_{\mc{E}_{ij,kl}}^A\otimes \wt{J}_{\Theta[\mc{E}_{ij,kl}]}^B.
\end{align}
Similar to the identity for the Choi operators of linear maps, we have the following identity for the Choi operators of supermaps
\begin{align}
    \wt{J}_{\Theta[\N]} = \tr_{A}\left[\wt{J}_\Theta^{AB}\left((\wt{J}_\N^T)^A\otimes\mathbb{I}^B\right)\right].
\end{align} 
From this expression, we note that the Choi operator $\wt{J}_\Theta$ can be interpreted as the Choi operator of the superoperator $\Delta_\Theta^{A\to B}$ defined by
\begin{align}
    \wt{\Delta}_\Theta^{A\to B}(\wt{J}_\N^A) = \wt{J}_{\Theta[\N]}^B,
\end{align}
so that $\wt{\Delta}_\Theta$ can be thought of as the superoperator corresponding to $\Theta$ that acts on the Choi operator $\wt{J}_\N$ rather than the superoperator $\N$. Or, we can consider the normalized version
\begin{align}
    \Delta_\Theta^{A\to B}(J_\N^A) = J_{\Theta[\N]}^B,
\end{align}

We now define a superchannel:
\begin{definition}[\cite{gour2019comparison}]
    \begin{enumerate}
        \item[{}]
        \item A supermap $\Theta$ is CP-preserving (CPP) if $\Theta[\N]$ is CP for every $\N$ that is CP.
        \item A supermap $\Theta$ is completely CPP if $\1^R\otimes\Theta^{A\to B}$ is CPP for every dynamical system $R=R_0\to R_1$.
        \item A supermap $\Theta$ is TP-preserving (TPP) if $\Theta[\N]$ is TP for every $\N$ that is TP.
        \item A superchannel is a completely CPP and TPP supermap.
    \end{enumerate}
\end{definition}
The following theorem can be established.
\begin{theorem}[\cite{gour2019comparison}, Thm. 1]
    A supermap $\Theta$ is a completely CPP if and only if $\Delta_\Theta$ is CP.
\end{theorem}

\begin{definition}[\cite{burniston2020necessary}]
    \begin{enumerate}
        \item[{}]
        \item A supermap $\Theta$ is CPTNI-preserving if $\Theta$ is CPP and 
        \begin{align}
            \tr[\Theta[\N](\rho)] \leq 1
        \end{align}
        for all CPTP $\N$ and all density operators $\rho$.
        \item A supermap $\Theta$ is completely CPTNI-preserving if $\1^R\otimes\Theta^{A\to B}$ is CPTNI-preserving for every dynamical system $R=R_0\to R_1$.
        \item A superinstrument $\{\Theta_x\}_x$ is a set of completely CPTNI-preserving supermaps such that $\sum_x \Theta_x$ is a superchannel.
    \end{enumerate}
\end{definition}
We remark here that it is crucial to require each $\Theta_x$ to be \textit{completely} CPTNI-preserving rather than simply CPTNI-preserving. This is because only completely CPTNI-preserving supermaps can sum to a superchannel (see Ref. \cite{burniston2020necessary}, Thm 4.1).

Finally, we reiterate the following definition of an inner product of two CP maps in terms of the Choi states:
\begin{align}
    \langle\mc{N},\mc{M}\rangle = \tr(J_\mc{N}J_\mc{M}).
\end{align}
The adjoint supermap of $\Theta$ relative to this inner product, denoted by $\Theta^\dagger$, is defined by:
\begin{align}
    \langle\Theta^\dagger(\N),\M\rangle = \langle\N,\Theta(\M)\rangle.
\end{align}

\section{Projective Robustness of a Dynamical Resource}
We will make heavy use of the projective robustness measure introduced in Ref. \cite{regula2022probabilistic,regula2022tight} for static resource theories, i.e., resource theories of quantum states. We will provide a brief review here and generalize the concept to dynamical resource theories, i.e., resource theories of quantum channels.

In a static resource theory with free states $\mbb{F}_S$, the projective robustness is defined as 
\begin{align}
    \Omega(\rho) = \inf \{\gamma: \rho\leq\lambda\sigma\leq\gamma\rho:\sigma\in\mathbb{F}_S\}.
\end{align}
Given the set of free states $\mathbb{F}_S$, the largest possible set of free operations is the set of resource non-generating (RNG) operations, which is the set of operations that satisfy the minimal requirement of a free operation. In other words, $\mathrm{RNG} =\{\mc{O}:\mc{O}(\rho)\in\mathbb{F}_S\ \forall\rho\in\mathbb{F}_S\}$. The projective robustness is useful because it provides a necessary condition for the probabilistic transformation of quantum resources using RNG. Specifically, in any resource theory with free states $\mathbb{F}_S$ and free operations $\mathbb{F}_C\subseteq\mathrm{RNG}$, if the transformation $\rho\mapsto\rho'$ can be accomplished probabilistically, then 
\begin{align}
    \Omega(\rho)\geq\Omega(\rho').
\end{align}
Incidentally, this becomes a necessary and sufficient condition when the free states form an affine subspace and $\mathbb{F}_C=\mathrm{RNG}$, i.e.,
\begin{align}
    \mathbb{F}_S = \mathrm{aff}(\mathbb{F}_S)\cap\mc{D}
\end{align}
where $\mathrm{aff}(\mathbb{F}_S)$ is the affine hull of $\mathbb{F}_S$, and $\mc{D}$ is the set of density operators. In other words, the transformation $\rho\mapsto\rho'$ can be achieved probabilistically using an RNG operation if and only if $\Omega(\rho)\geq\Omega(\rho')$.

For our purposes, we would like to extend the necessary condition to dynamical resource theories. We define the projective robustness of a dynamical resource as
\begin{align}
    \Omega(\mc{N}) = \inf\{\gamma:\N\leq_{c.p.}\lambda\mc{M}\leq_{c.p.}\gamma\N,\ \mc{M}\in\mathbb{F}_C\}
\end{align}
where $\leq_{c.p.}$ denotes completely positive ordering, in the sense that $\mc{A}\leq_{c.p.}\mc{B}$ if $\mc{B}-\mc{A}$ is a completely positive map. Equivalently, this can be phrased in terms of the Choi states:
\begin{align}
    \Omega(\mc{N}) = \inf\{\gamma:J_\N\leq\lambda J_\mc{M}\leq \gamma J_\N,\ \mc{M}\in\mathbb{F}_C\}.
\end{align}
Most of the results in Ref. \cite{regula2022tight} can be generalized to dynamical resource theories. For our case, we will only require the following results. The proofs are completely analogous to the ones in Ref. \cite{regula2022tight} and are reproduced here to be self-contained.
\begin{theorem}[cf. Ref. \cite{regula2022tight}, Thm. 4]\label{thm:proj-rob-nec-condition}
    If there exists a probabilistic free transformation $\mc{N}\mapsto\mc{N}'$, then $\Omega(\mc{N})\geq\Omega(\mc{N}')$.
\end{theorem}
\begin{proof}
    Let $\mc{M}$ be a free channel such that $\N\leq_{c.p.}\lambda\mc{M}\leq_{c.p.}\gamma\N$, and let $\Theta=\{\Theta_0,\Theta_1\}$ be a free probabilistic transformation such that $\Theta_0[\N]=p\N'$, or equivalently $\mc{N}'=\displaystyle\frac{\Theta_0[\N]}{\tr(J_{\Theta_0[\N]})}$. Since $\Theta_0[\mc{N}]\leq_{c.p.}\lambda\Theta_0[\mc{M}]$, we have
    \begin{align}
        \frac{\Theta_0[\N]}{\tr J_{\Theta_0[\N]}} \leq_{c.p.} \lambda\frac{\tr J_{\Theta_0[\mc{M}]}}{\tr J_{\Theta_0[\mc{N}]}}\frac{\Theta_0[\mc{M}]}{\tr J_{\Theta_0[\mc{M}]}}.
    \end{align}
    Similarly,
    \begin{align}
        \lambda\frac{\Theta_0[\mc{M}]}{\tr J_{\Theta_0[\mc{M}]}} \leq_{c.p.} \gamma\frac{\tr J_{\Theta_0[\mc{N}]}}{\tr J_{\Theta_0[\mc{M}]}}\frac{\Theta_0[\mc{N}]}{\tr J_{\Theta_0[\mc{N}]}}.
    \end{align}
    Since $\Theta$ is a free superchannel and $\mc{M}$ is a free channel, $\displaystyle\frac{\Theta_0[\mc{M}]}{\tr J_{\Theta_0}[\mc{M}]}\in\mathbb{F}_C$. So, from the fact that 
    \begin{align}
        \mc{N}' \leq_{c.p.} \lambda\frac{\tr J_{\Theta_0[\mc{M}]}}{\tr J_{\Theta_0[\mc{N}]}}\frac{\Theta_0[\mc{M}]}{\tr J_{\Theta_0[\mc{M}]}} \leq_{c.p.} \gamma \mc{N}'
    \end{align}
    we can conclude $\Omega(\N')\leq\gamma$. This is true for any $\gamma$ and free channel $\mc{M}$ such that $\mc{N}\leq_{c.p.}\lambda\mc{M}\leq_{c.p.}\gamma\mc{N}$, which means $\Omega(\mc{N}')\leq\Omega(\mc{N})$.
\end{proof}
\begin{remark}
    Similar to Ref. \cite{regula2022tight}, we can in fact prove this for when the transformation can only be done in the limit. That is, there exists a sequence of free superchannels $\Theta^{(n)}=\{\Theta_0^{(n)},\Theta_1^{(n)}\}$ such that $\lim_{n\to\infty} \displaystyle\frac{\Theta_0^{(n)}[\mc{N}]}{J_{\Theta^{(n)}_0[\mc{N}]}} = \N'$, but we will refrain from this discussion as it will be unnecessary for our case. 
\end{remark}
This result imply the following corollary on the achievable fidelity of probabilistic resource distillation. 
\begin{proposition}\label{prop:fidelity-bound}
    If there exists a free probabilistic transformation $\mc{N}\to\N'$ such that $\langle\mc{N}',\Phi\rangle = F$ for some $\N'$, then 
    \begin{align}
        F \leq \frac{\Omega(\N)}{\Omega(\N)+(1-F_{\mbb{F}_C}(\Phi))/F_{\mbb{F}_C}(\Phi)}.
    \end{align}
    where $F_{\mbb{F}_C}(\Phi) = \max_{\mc{M}\in\mbb{F}_C}\langle \mc{M},\Phi\rangle$ is the largest overlap between free channels and $\Phi$.
\end{proposition}
\begin{proof}
    Let $\mc{M}$ be a free channel such that $\Omega(\N')=\gamma$ and $\mc{N}'\leq\lambda\mc{M}\leq\gamma\mc{N}'$. Then,
    \begin{align}
        &F = \langle\mc{N}',\Phi\rangle \leq \lambda\langle\mc{M},\Phi\rangle \leq \lambda F_{\mbb{F}_C}(\Phi), \\
        &1-F_{\mbb{F}_C}(\Phi) \leq \langle\mc{M},d^2\Delta-\Phi\rangle \leq \frac{\gamma}{\lambda}\langle\N',d^2\Delta-\Phi\rangle = \frac{\gamma}{\lambda}(1-F),
    \end{align}
    where again $\Delta$ is the completely depolarizing channel. Multiplying the two inequalities, we conclude 
    \begin{align}
        \Omega(\N) \geq \Omega(\N') = \gamma \geq \frac{F(1-F_{\mbb{F}_C}(\Phi))}{(1-F)F_{\mbb{F}_C}(\Phi)}.
    \end{align}
    This proves the claim upon rearranging the terms.
\end{proof}

\begin{proposition}\label{prop:submultiplicativity}
    The projective robustness is sub-multiplicative. That is, $\Omega(\mc{N}_1\otimes\mc{N}_2)\leq\Omega(\mc{N}_1)\Omega(\mc{N}_2)$.
\end{proposition}
\begin{proof}
    For any feasible $\mc{M}_1$ and $\mc{M}_2$ such that $\N_1\leq_{c.p.}\lambda_1\mc{M}_1\leq_{c.p.}\gamma_1\mc{N}_1$ and $\N_2\leq_{c.p.}\lambda_2\mc{M}_2\leq_{c.p.}\gamma_2\mc{N}_2$, we have $\mc{N}_1\otimes\mc{N}_2\leq_{c.p.}\lambda_1\lambda_2\mc{M}_1\otimes\mc{M}_2\leq_{c.p.}\gamma_1\gamma_2\mc{N}_1\otimes\mc{N}_2$. In other words, $\mc{M}_1\otimes\mc{M}_2$ is a feasible solution and attains $\gamma_1\gamma_2$. Therefore, $\Omega(\mc{N}_1\otimes\mc{N}_2)\leq\Omega(\mc{N}_1)\Omega(\mc{N}_2)$.
\end{proof}

\begin{theorem}[cf. Ref. \cite{regula2022tight}, Prop. 18]
    In any convex resource theory, the probability of transforming $\N$ into $\N'$ with RNG superinstruments is upper bounded by $\Xi(\N,V_\mathbb{F}(\N')^{-1})$, where
    \begin{align}
        \Xi(\N,V_\mathbb{F}(\N')^{-1}) = \max \quad& \tr ~WJ_\N \\
        \mathrm{s.t.} \quad& 0 \leq W \leq Z \leq d\sigma\otimes\mathbb{I},\ \sigma\in\mc{D} \\
        &\tr~WJ_\N = \tr~ZJ_\N \\
        & \tr~WJ_\mc{M} \leq V_{\mathbb{F}_C(\N')} \cdot \tr~ZJ_\mc{M} \quad\forall \mc{M}\in\mathbb{F}_C
    \end{align}
    and $V_{\mathbb{F}_C}(\N) = \max_{\mc{M}\in\mathbb{F}_C} \tr(\Pi_{J_\mc{N}}J_\mc{M})$ with $\Pi_{J_\N}$ being the projection onto the support of $J_\N$.
\end{theorem}
\begin{proof}
    Assume there exists a completely CPTNI-preserving supermap $\Theta_0$ such that $\Theta_0[\N] = p\N'$ for some $p\in(0,1]$. Then,
    \begin{align}
        &p = \tr J_{\Theta_0[\N]} = \tr\Delta_{\Theta_0}(J_\N) = \tr\left[\Delta_{\Theta_0}^\dagger(\mbb{I})J_\N\right], \\
        &p = \tr\left(\Pi_{J_{\N'}} J_{\Theta_0[\N]}\right) = \tr\left[\Pi_{J_{\N'}}\Delta_{\Theta_0}(J_\N)\right] = \tr\left[\Delta_{\Theta_0}^\dagger(\Pi_{J_{\N'}})J_\N\right].
    \end{align}
    where $\Delta_{\Theta_0}^\dagger$ is the dual of $\Delta_{\Theta_0}$. Since $\Theta_0$ is completely CPP, $\Delta_{\Theta_0}$ is completely positive, and therefore $\Delta_{\Theta_0}^\dagger$ is completely positive. Furthermore, for any CPTP superoperator $\mc{O}$, 
    \begin{align}
        \tr\left[\Delta_{\Theta_0}^\dagger(\mbb{I})J_\mc{O}\right] = \tr\left[\Delta_{\Theta_0}(J_\mc{O})\right] = \tr J_{\Theta_0[\mc{O}]} \leq \tr J_\mc{O},
    \end{align}
    which means that $\tr\left\{\left[\mbb{I}-\Delta_{\Theta_0}^\dagger(\mbb{I})\right]J_\mc{O}\right\} \geq 0$ for all $J_\mc{O}\geq0$ with such that $\tr_{A_1}J_\mc{O}^A=\mbb{I}/d$. This implies that $\Delta_{\Theta_0}^\dagger(\mbb{I})\leq d\sigma\otimes\mbb{I}$ with $0\leq\sigma\leq\mbb{I}$ must hold. We thus conclude $\Delta_{\Theta_0}^\dagger(\Pi_{J_{\N'}}) \leq \Delta_{\Theta_0}^\dagger(\mathbb{I}) \leq d\sigma\otimes\mbb{I}$. Now, for any $\mc{M}\in\mbb{F}_C$, 
    \begin{align}
        \tr\left[\Delta_{\Theta_0}^\dagger(\Pi_{J_{\N'}})J_{\mc{M}}\right] &= \tr\left[\Pi_{J_{\N'}}\Delta_{\Theta_0}(J_\mc{M})\right] \\
        &= \tr\Delta_{\Theta_0}(J_\mc{M}) \cdot \tr\left[\Pi_{J_{\N'}}\frac{\Delta_{\Theta_0}(J_\mc{M})}{\tr\Delta_{\Theta_0}(J_\mc{M})}\right] \\
        &\leq V_{\mbb{F}_C}(\N') \cdot \tr\Delta_{\Theta_0}(J_\mc{M}) \\
        &= V_{\mbb{F}_C}(\N') \cdot \tr\left[\Delta_{\Theta_0}^\dagger(\mbb{I})J_\mc{M}\right].
    \end{align}
    Setting $W = \Delta_{\Theta_0}^\dagger(\Pi_{J_{\N'}})$ and $Z = \Delta_{\Theta_0}^\dagger(\mbb{I})$ as optimization variables proves the theorem.
\end{proof}
\begin{proposition}\label{prop:resource-distill-prob-ub}
    In any convex resource theory, the probability of transforming $\N$ into a pure target channel $\Phi$ with fidelity $F$, i.e., $J_\Phi$ is pure, $\Theta_0[\N]=p\N'$ and $\langle\Phi,\N'\rangle\geq F$ where $\Theta=\{\Theta_0,\Theta_1\}$ is an RNG superinstrument is upper bounded by
    \begin{align}
        \Xi_\epsilon(\N,V_{\mbb{F}_C}(\N')^{-1}) = \max \quad& \tr(ZJ_\N) \\
        \text{s.t.} \quad& 0 \leq W \leq Z \leq d\sigma\otimes\mbb{I}, \ \sigma\in\mc{D} \\
        & \tr(WJ_\N) = F\tr(ZJ_\N) \\
        & \tr(WJ_{\mc{M}}) \leq V_{\mathbb{F}_C}(\Phi)\cdot\tr(ZJ_\mc{M}) \quad\forall\mc{M}\in\mathbb{F}_C.
    \end{align}
\end{proposition}
\begin{proof}
    We first note that we can relax the condition $\tr(WJ_\N) = F\tr(ZJ_\N)$ to $\tr(WJ_N) \geq F\tr(ZJ_\N)$. This is because if there exist $W$ and $Z$ such that $\tr(WJ_\N)>F\tr(ZJ_\N)$, then we can rescale $W$ to $W'$ such that $\tr(W'J_\N)=F\tr(ZJ_\N)$. This new choice of $W'$ and $Z$ is still feasible, and attains the same objective value.

    For any CPTNI-preserving supermap $\Theta_0$ such that $\Theta_0[\N] = p\N'$ where $\langle\Phi,\N'\rangle\geq F$, we have
    \begin{align}
        & p = \tr(J_{\Theta_0[\N]}) =  \tr(\mbb{I}\Delta_{\Theta_0}(J_\N)) = \tr\left(\Delta_{\Theta_0}^\dagger(\mbb{I})J_\N\right), \\
        &\tr\left(\Delta_{\Theta_0}^\dagger(J_\Phi)J_\N\right) = \tr\left(J_\Phi\Delta_{\Theta_0}(J_\N)\right) = \tr\left(J_\Phi J_{\Theta_0[\N]}\right) = \langle\Phi,p\N'\rangle \geq pF = \tr(\Delta_{\Theta_0}^\dagger(\mbb{I}) J_\N)F  
    \end{align}
    Moreover,
    \begin{align}
        \tr\left(\Delta_{\Theta_0}^\dagger(\Phi)J_\mc{M}\right) &= \tr(\Delta_{\Theta_0}(J_\mc{M}))\frac{\tr J_\Phi \Delta_{\Theta_0}(J_\mc{M})}{\tr(\Delta_{\Theta_0}(J_\mc{M}))} \\
        &= \tr(\Delta_{\Theta_0}(J_\mc{M}))\frac{\tr J_\Phi J_{\Theta_0[\mc{M}]}}{\tr(J_{\Theta_0[\mc{M}]})} \\
        &\leq \tr(\Delta_{\Theta_0}(J_\mc{M})) V_{\mbb{F}_C}(\Phi) \\
        &= V_{\mbb{F}_C}(\Phi) \cdot \tr\left(\Delta_{\Theta_0}^\dagger(\mbb{I})J_\mc{M}\right).
    \end{align}
    Therefore, $W=\Delta_{\Theta_0}^\dagger(\Phi)$ and $Z=\Delta_{\Theta_0}^\dagger(\mbb{I})$ are feasible solutions. Optimizing over such $W$ and $Z$ yields the upper bound.
\end{proof}

\section{Equivalence to Probabilistic Distillation of Quantum Channels}
In this section, we will prove the claim made in the main text that $\Theta_0$ can be chosen without loss of generality such that
\begin{align}
    \tr\Theta_0[\mc{N}_1\otimes\cdots\otimes\mc{N}_n](\rho)=p(\Theta_0;\mc{N}_1,\cdots,\mc{N}_n) \quad \forall\,\rho.
\end{align}
\begin{proof}
    Suppose $\Theta_0$ is a supermap such that 
    \begin{align}
        \inf_\psi F\left(\frac{\mc{O}^B(\psi^{RB_0})}{\tr\mc{O}^B(\psi^{RB_0})},\psi^{RB_0}\right) > F(\mc{N}_i) \quad\forall i.
    \end{align}
    where $\mc{O} = \Theta_0[\N_1\otimes\cdots\otimes\N_n]$ is a CPTNI map and not necessarily a sub-normalized quantum channel. Suppose we perform channel twirling on $\mc{O}$, then
    \begin{align}
        \mc{T}^{B\to B}[\mc{O}^B](\psi^{RB_0}) = p\mc{D}_F^B(\psi^{RB_0}),
    \end{align}
    where $p = \tr\left(\mc{O}^B(\phi_+^{RB_0})\right)$ and $pF = \tr\left(\mc{O}^B(\phi_+^{RB_0})\phi_+^{RB_1}\right)$. This can be derived by invoking the Choi-Jami{\l}kowski isomorphism, where channel twirling becomes $(U\otimes U^*)$-twirling on the Choi operator $J_\mc{O}$. So,
    \begin{align}
        F\left(\frac{\mc{T}^{B\to B}[\mc{O}^B](\psi^{RB_0})}{\tr~\mc{T}^{B\to B}[\mc{O}^B](\psi^{RB_0})},\psi^{RB_0}\right) &= F(\mc{D}_F^B(\psi^{RB_0}),\psi^{RB_0}) \\
        &= F\left(F\psi^{RB_0}+(1-F)\frac{d\psi^R\otimes\mbb{I}^{B_0}-\psi^{RB_0}}{d^2-1},\psi^{RB_0}\right) \\
        &= F + (1-F)\frac{d\tr(\psi^R)^2-1}{d^2-1} \\
        &\geq F.
    \end{align}
    But we also have
    \begin{align}
        F\left(\frac{\mc{O}^B(\phi_+^{RB_0})}{\tr~\mc{O}^B(\phi_+^{RB_0})},\phi_+^{RB_0}\right) = F.
    \end{align}
    Therefore, 
    \begin{align}
        \inf_\psi F\left(\frac{\mc{T}^{B\to B}[\mc{O}^B](\psi^{RB_0})}{\tr~\mc{T}^{B\to B}[\mc{O}^B](\psi^{RB_0})},\psi^{RB_0}\right) &\geq F\left(\frac{\mc{O}^B(\phi_+^{RB_0})}{\tr~\mc{O}^B(\phi_+^{RB_0})},\phi_+^{RB_0}\right) \\ &\geq \inf_\psi F\left(\frac{\mc{O}^B(\psi^{RB_0})}{\tr\mc{O}^B(\psi^{RB_0})},\psi^{RB_0}\right) .
    \end{align}
    This means that performing channel twirling cannot decrease the worst-case fidelity. Therefore, we can without loss of generality take $J_{\Theta_0}$ to be invariant under $U\otimes U^*$ twirling on $B$, so that it has the form
    \begin{align}
        J_{\Theta_0}=X^{A^{\otimes n}}\otimes\phi_+^B + Y^{A^{\otimes n}}\otimes\mbb{I}^B
    \end{align}
    The action of such a supermap is 
    \begin{align}
        \Theta_0[\mc{N}_1\otimes\cdots\otimes\mc{N}_n] = \tr \left[\left(J_{\mc{N}_1}\otimes\cdots\otimes J_{\mc{N}_n}\right)X^T\right]\mathrm{id} + \tr \left[\left(J_{\mc{N}_1}\otimes\cdots\otimes J_{\mc{N}_n}\right)Y^T\right]\Delta,
    \end{align}
    where $\Delta(\rho)=\mbb{I}/d$ is the completely depolarizing channel. Such a supermap satisfies that $\tr\Theta_0[\mc{N}_1\otimes\cdots\otimes\mc{N}_n](\rho)$ only depends on the channels $\mc{N}_i$ and is fixed for all $\rho$.
\end{proof}

\section{Proof of Proposition \ref{prop:optimality}}
Proposition \ref{prop:optimality} can be proven by computing the projective robustness of a generalized Pauli channel and invoking Propositions \ref{prop:fidelity-bound} and \ref{prop:submultiplicativity}. In the resource theory of communication, the free channels are replacement channels. Therefore, the projective robustness of a generalized Pauli channel is 
\begin{align}
    \Omega(\mc{N}_i)=\inf \left\{\gamma:J_{\mc{N}_i}\leq\lambda \frac{\mathbb{I}}{d}\otimes\sigma\leq\gamma J_{\mc{N}_i},\ \sigma\in\mc{D}\right\}.
\end{align}
Since $J_{\mc{N}_i}$ is invariant under generalized Pauli twirling $\mc{T}_P(\cdot) = \sum_{x,z} \frac{1}{d^2}U_{x,z}\otimes U_{x,z}^*(\cdot)U_{x,z}^\dagger\otimes U_{x,z}^T$, we can without loss of generality take $\sigma$ to be the maximally mixed state. Therefore, 
\begin{align}
     \Omega(\mc{N}_i) &= \inf\{\gamma:J_{\mc{N}_i}\leq\lambda\mathbb{I}/d^2\leq\gamma J_{\mc{N}_i}\}, \\
     &= \frac{p_{max}^{(i)}}{p_{min}^{(i)}}
\end{align}
Moreover, 
\begin{align}
    F_{\mbb{F}_C}(\mathrm{id}) = \max_{\mc{M}\in\mbb{F}_C}\langle\mc{M},\mathrm{id}\rangle = \max_{\mc{M}\in\mbb{F}_C}\tr\left(\left(\frac{\mbb{I}}{d}\otimes\sigma\right)\phi_+\right) = \frac{1}{d^2}.
\end{align}
Using Prop. \ref{prop:fidelity-bound}, $\langle\N_{out},\mathrm{id}\rangle$, which is precisely the entanglement fidelity of $\N_{out}$, must satisfy
\begin{align}
    F_e(\N_{out}) &\leq \frac{\Omega\left(\bigotimes_i\mc{N}_i\right)}{\Omega\left(\bigotimes_i\mc{N}_i\right)+ \frac{1-1/d^2}{1/d^2}} \\
    &\leq \frac{\prod_i \Omega(\N_i)}{\prod_i\Omega(\N_i)+(d^2-1)} \\
    &= \frac{\prod_i p_{max}^{(i)}} {\prod_i p_{max}^{(i)}+(d^2-1)\prod_i p_{min}^{(i)}}.
\end{align}
Finally, $F(\N_{out})\leq F_e(\N_{out})$. This establishes the claim.

\section{Proof of Prop. \ref{prop:optimal-prob}}
As discussed previously, the largest probability of converting a channel $\N^{AA'}$ to $\N'^B$ such that $\langle\mathrm{id},\N'\rangle=F'$ using NS super-instruments is
\begin{align}
    \max \quad& \tr(ZJ_\N) \\
    \text{subject to} \quad& 0 \leq W \leq Z \leq d\sigma\otimes\mbb{I} \\
    & \tr(WJ_\N) = F'\tr(ZJ_\N) \\
    & \tr\left[W(\mbb{I}^{A_0A_0'}\otimes\sigma^{A_1A_1'})\right] \leq V_{\mathbb{F}_C}(\mathrm{id})\cdot\tr\left[Z(\mbb{I}^{A_0A_0'}\otimes\sigma^{A_1A_1'})\right] \quad\forall\sigma\in\mc{D}.
\end{align}
The last constraint is equivalent to $\tr_{A_0A_0'}W \leq \frac{1}{d^2}\tr_{A_0A_0'} Z$, where $d$ is the dimension of the input and output of the identity channel. Since we are interested in the case where $\N^{AA'}=\mc{D}_{F_1}^A\otimes\mc{D}_{F_2}^{A'}$, the Choi operator is invariant under twirling on $A_0A_1$ and on $A_0'A_1'$, and so we can without loss of generality take $W$ and $Z$ to be of the form
\begin{align}
    &W = w_1 \phi_+^A \otimes \phi_+^{A'} + w_2 \phi_+^A \otimes (\phi_+^\perp)^{A'} + w_3 (\phi_+^\perp)^A \otimes \phi_+^{A'} + w_4 (\phi_+^\perp)^A \otimes (\phi_+^\perp)^{A'} \\
    &Z = z_1 \phi_+^A \otimes \phi_+^{A'} + z_2 \phi_+^A \otimes (\phi_+^\perp)^{A'} + z_3 (\phi_+^\perp)^A \otimes \phi_+^{A'} + z_4 (\phi_+^\perp)^A \otimes (\phi_+^\perp)^{A'}
\end{align}
The SDP thus reduces to a linear program:
\begin{align}
    \max \quad& \vec{f}\cdot\vec{z} \\
    \text{subject to} \quad& 0 \leq \vec{w} \leq \vec{z} \leq 1 \\
    & \vec{f}\cdot\vec{w} = \frac{F_1F_2}{F_1F_2+\frac{1}{d^2-1}(1-F_1)(1-F_2)} \vec{f}\cdot\vec{z} \\
    & \vec{a}\cdot\vec{w} \leq \frac{1}{d^2} \vec{a}\cdot\vec{z}
\end{align}
where $\vec{f} = (F_1F_2,F_1(1-F_2),(1-F_1)F_2,(1-F_1)(1-F_2))$ and $\vec{a} = (1,d^2-1,d^2-1,(d^2-1)^2)$. The dual of the linear program is
\begin{align}
    \min \quad& \vec{t}\cdot\vec{1} \\
    \text{subject to} \quad& \vec{f} + \vec{s} - \vec{t} - \mu\frac{F_1F_2}{F_1F_2+\frac{1}{d^2-1}(1-F_1)(1-F_2)}\vec{f} + \frac{\eta}{d^2}\vec{a} = 0 \\
    & \vec{r} - \vec{s} + \mu\vec{f} - \eta\vec{a} = 0 \\
    & \vec{r},\vec{s},\vec{t} \geq0 \\
    &\eta \geq 0.
\end{align}
One can verify that the following is a feasible solution for the dual program when $F_1,F_2>1/d^2$.
\begin{align}
    &\vec{r} = 
    \begin{pmatrix}
      0 \\
      \frac{\eta(d^2-1)^2}{d^2}\left[1-\frac{F_1(1-F_2)}{(d^2-1)F_1F_2}\right]-\frac{pF_1(1-F_2)}{F_1F_2} \\
      \frac{\eta(d^2-1)^2}{d^2}\left[1-\frac{(1-F_1)F_2}{(d^2-1)F_1F_2}\right]-\frac{p(1-F_1)F_2}{F_1F_2} \\
      \frac{\eta(d^2-1)^2}{d^2}\left[1-\frac{p}{d^2F_1F_2}\right]-\frac{p(1-F_1)(1-F_2)}{F_1F_2}
    \end{pmatrix},
    \quad \vec{s} = 
    \begin{pmatrix}
        \mu F_1F_2 - \eta \\
        \frac{\eta(d^2-1)}{d^2}\left[\frac{(d^2-1)F_1(1-F_2)}{(1-F_1)(1-F_2)}-1\right] \\
        \frac{\eta(d^2-1)}{d^2}\left[\frac{(d^2-1)(1-F_1)F_2}{(1-F_1)(1-F_2)}-1\right] \\
        0
    \end{pmatrix},
    \quad \vec{t} = 
    \begin{pmatrix}
        p \\ 0 \\ 0 \\ 0
    \end{pmatrix},
    \\
    &\mu = \frac{p}{F_1F_2}\left(1+\frac{(d^2-1)^2\eta}{d^2(1-F_1)(1-F_2)}\right), \\
    &\eta = \max\left\{\frac{d^2pF_1(1-F_2)}{(d^2-1)[(d^2-1)F_1F_2-F_1(1-F_2)]}, \frac{d^2p(1-F_1)F_2}{(d^2-1)[(d^2-1)F_1F_2-(1-F_1)F_2]},  \frac{d^2p(1-F_1)(1-F_2)}{(d^2-1)[(d^2-1)F_1F_2-p]}\right\},
\end{align}
where 
\begin{align}
    p = F_1F_2+\frac{1}{d^2-1}(1-F_1)(1-F_2).
\end{align}
This proves that the success probability is upper bounded by $\vec{t}\cdot\vec{1} = F_1F_2+\frac{1}{d^2-1}(1-F_1)(1-F_2)$.

\section{Proof of Proposition \ref{prop:classical-optimality}}
This is again proven by using the projective robustness. The projective robustness of a classical channel $\mc{P}=\sum_{x,y} \mc{P}(y|x)\op{y}{x}W\op{x}{y}$ is
\begin{align}
    \Omega(\mc{P}) = \inf\left\{\gamma: J_{\mc{P}} \leq \lambda(\mbb{I}/d)\otimes\sigma \leq \gamma J_{\mc{P}},\ \sigma=\sum_y\sigma_y\state{y},\ \sigma_y\geq0,\ \sum_y \sigma_y=1\right\}.
\end{align}
where the Choi operator $J_{\mc{P}}=\sum_{x,y} \mc{P}(y|x)\state{x}\otimes\state{y}$. For each $y$, the constraint requires that 
\begin{align}
    \mc{P}(y|x) \leq \lambda\sigma_y/d \leq \gamma\mc{P}(y|x) \quad\forall x
\end{align}
So the optimal choice of $\sigma_y$ is such that $\lambda\sigma_y/d=\max_x \mc{P}(y|x)$, and the optimal $\gamma$ is $\displaystyle\frac{\max_x \mc{P}(y|x)}{\min_x \mc{P}(y|x)}$. Since the constraint must hold for each $y$, we conclude that the optimal $\gamma$ is
\begin{align}
    \Omega(\mc{P}) = \max_y \frac{\max_x \mc{P}(y|x)}{\min_x \mc{P}(y|x)}.
\end{align}
Now, consider an output channel $\mc{P}_{out}$ with $F(\mc{P}_{out})=\min_x\mc{P}(x|x)=F$. For every $y$, we have $\max_x \mc{P}_{out}(y|x) \geq \mc{P}_{out}(y|y) \geq F$. On the other hand, since $\mc{P}_{out}(x|x) \geq F$ for every $x$, there exists $y^*$ such that $\mc{P}_{out}(y^*|x) \leq \frac{1-F}{d-1}$. Therefore, 
\begin{align}
    \Omega(\mc{P}_{out}) &= \max_y \frac{\max_x \mc{P}_{out}(y|x)}{\min_x \mc{P}_{out}(y|x)} \\
    &\geq \frac{\max_x \mc{P}_{out}(y^*|x)}{\min_x \mc{P}_{out}(y^*|x)} \\
    &\geq \frac{F}{1-F}(d-1).
\end{align}
Thm. \ref{thm:proj-rob-nec-condition} thus requires $\Omega(\mc{P})\geq\Omega(\mc{P}_{out})\geq\frac{F}{1-F}(d-1)$. Rearranging the terms proves the result.

\section{Proof of Theorem \ref{thm:general-rt-suff-condition}}
To prove Thm \ref{thm:general-rt-suff-condition}, we will first require a simple lemma:
\begin{lemma}\label{lemma:Rg=Rs}
    In any convex resource theory with a free generalized channel twirling, $R_s(\Phi)=R_g(\Phi)$ if and only if
    \begin{align}
        \{F:\N_F=F\Phi+(1-F)\Lambda \in \mbb{F}_C\} = [0,F_c] = [0,R_g(\Phi)^{-1}].
    \end{align}
\end{lemma}
\begin{proof}
    In any resource theory with a free generalized twirling, for any feasible $\M\in\mbb{F}_C$ such that $\Phi\leq_{c.p.}\lambda\M$, it holds that $\Phi\leq_{c.p.}\lambda\mc{T}[\M]=\lambda(F\Phi+(1-F)\Lambda)$ for some $F$. Clearly, the minimal $\lambda$ is attained for the largest $F$ such that $F\Phi+(1-F)\Lambda$ is free. It is not hard to see that the largest such $F$ is $F_c = \max_{\M\in\mbb{F}_C}\langle\Phi,\M\rangle$. Therefore, $R_g(\Phi)=1/F_c$.

    Let us now prove the only if direction. If $R_s(\Phi)=R_g(\Phi)$, then there exist an optimal  $\M^*$ such that $\Phi\leq_{\mbb{F}_C}R_g(\Phi)\M^*$. From this, we conclude that $\Phi \leq_{\mbb{F}_C} R_g(\Phi)\mc{T}[\M^*] \allowbreak= R_g(\Phi)(F\Phi+(1-F)\Lambda) = (F\Phi+(1-F)\Lambda)/F_c$ for some $F\leq F_c$ (since $F_c$ is the largest $F$ such that the isotropic channel $\N_F$ is free). This is only possible if $F=F_c$ and $\Lambda$ is free. By convexity, every isotropic channel $\N_F$ with $F\in[0,F_c]$ is free.

    Conversely, if free isotropic channels consists of $\N_F$ with $F\in[0,F_c]$, then we solve $R_s(\Phi)=\min_{\M\in\mbb{F}_C}\{\lambda:\Phi\leq_{\mbb{F}_C}\lambda\M\}$ by taking, without loss of generality, $M$ to be isotropic channels. This yields $R_s(\Phi)=1/F_c$ and proves the lemma.
\end{proof}

\begin{proof}[Proof of Thm. \ref{thm:general-rt-suff-condition}]
    We first prove $\Theta$ is RNG for when the resource theory satisfies conditions 1 and 2(b). Because the isotropic channel is invariant under twirling, the optimal $\M^*$ in the optimization program for $\Omega(\N_F)$ can be chosen to be an isotropic channel. Then, it is not hard to see that the optimal $\lambda^*\M^*$ is
    \begin{align}
        \lambda^*\mc{M}^* = F\Phi+\frac{F}{F_c}(1-F_c)\Lambda = \frac{F}{F_c}\N_{F_c},
    \end{align}
    and $\Omega(\N_F)=\frac{F}{1-F}\frac{1-F_c}{F_c}$. If $\Omega(\N_F\otimes\N_F)=\Omega(\N_F)^2$, then we can take the optimal solution for $\Omega(\N_F^{\otimes 2})$ to be $\mc{M}^*\otimes\mc{M}^*$. Consider the dual program of $\Omega(\N_F\otimes\N_F)$ 
    \begin{align}
        \sup\{\langle\mc{A}, \N_F\otimes\N_F\rangle:\  \langle\mc{B},\N_F\otimes\N_F\rangle=1,\ \langle\mc{B}-\mc{A},\M\rangle\geq0 \ \forall\M\in\mathbb{F}_C\} \label{eq:dual}
    \end{align}
    with dual variables $\mc{A}$ and $\mc{B}$, both of which are CP. Notice that if $\mc{A}^*$ and $\mc{B}^*$ are optimal solutions, then $\mc{T}^\dagger\otimes\mc{T}^\dagger[\mc{A}^*]$ and $\mc{T}^\dagger\otimes\mc{T}^\dagger[\mc{B}^*]$ are also optimal. This is because $\N_F$ is invariant under channel twirling and
    \begin{align}
        \langle\mc{T}^\dagger\otimes\mc{T}^\dagger[\mc{B}^*]-\mc{T}^\dagger\otimes\mc{T}^\dagger[\mc{A}^*],\M \rangle = \langle\mc{B}^*-\mc{A}^*,\mc{T}\otimes\mc{T}[\M]\rangle \geq 0 \qquad\forall\M\in\mathbb{F}_C.
    \end{align}
    It can be readily verified that
    \begin{align}
        \mc{T}^\dagger[\N] = \langle\N,\Phi\rangle\mc{P}+\langle\N,\Lambda\rangle(d^2\Delta-\mc{P}),
    \end{align}
    where $\Delta$ is the completely depolarizing channel and $\mc{P}$ is the CP map whose corresponding Choi operator is $P$. This means that we can take $\mc{A}^*$ and $\mc{B}^*$ to be of the form $a_1\mc{P}\otimes\mc{P} + a_2\mc{P}\otimes(d^2\Delta-\mc{P}) + a_3(d^2\Delta-\mc{P})\otimes\mc{P} + a_4(d^2\Delta-\mc{P})\otimes(d^2\Delta-\mc{P})$. By complementary slackness, the optimal dual variables must satisfy 
    \begin{align}
        &\langle\mc{A}^*,\ \M^*\otimes\M^*-\N_F\otimes\N_F\rangle = 0, \\
        &\langle\mc{B}^*,\ \Omega(\N_F\otimes\N_F)\N_F\otimes\N_F-\M^*\otimes\M^*\rangle = 0.
    \end{align}
    By direct calculation, we see that $\M^*\otimes\M^*-\N_F\otimes\N_F$ has nonzero terms of $\Phi\otimes\Lambda$, $\Lambda\otimes\phi$, and $\Lambda\otimes\Lambda$, while $\Omega(\N_F\otimes\N_F)\N_F\otimes\N_F-\M^*\otimes\M^*$ has nonzero terms of $\Phi\otimes\Phi$, $\Phi\otimes\Lambda$, and $\Lambda\otimes\Phi$. This implies $\mc{A}^*$ must necessarily be of the form $x\mc{P}\otimes\mc{P}$, and similarly $\mc{B}^*$ must have the form $y(d^2\Delta-\mc{P})\otimes(d^2\Delta-\mc{P})$. Plugging $\mc{A}^*$ and $\mc{B}^*$ back into the conditions in Eq.~\eqref{eq:dual}, namely $\langle\mc{B},~\N_F\otimes\N_F\rangle=1$ and $\langle\mc{B}^*-\mc{A}^*,\M\rangle \geq 0$ for all $\M\in\mbb{F}_C$, which holds in particular for $\M=\N_{F_c}\otimes\N_{F_c}$, we conclude that $y = \frac{1}{(1-F)^2}$ and $x\leq\frac{(1-F_c)^2}{F_c^2(1-F)^2}$. On the other hand, strong duality clearly holds for $\Omega(\N_F\otimes\N_F)$. This is because we can choose $\M^*\otimes\M^*$ and sufficiently large $\lambda$ and $\gamma$, whence $\N_F\otimes\N_F < \lambda \M^*\otimes\M^* < \gamma\N_F\otimes\N_F$. Therefore, the dual optimal is precisely $\left(\frac{F}{1-F}\frac{1-F_c}{F_c}\right)^2$, and so $x$ must be $\frac{(1-F_c)^2}{F_c^2(1-F)^2}$. To summarize, the optimal solutions are
    \begin{align}
        &\mc{A}^* = \frac{(1-F_c)^2}{F_c^2(1-F)^2} \mc{P}\otimes\mc{P} \label{eq:optimal-A}\\
        &\mc{B}^* = \frac{1}{(1-F)^2}(d^2\Delta-\mc{P})\otimes(d^2\Delta-\mc{P}) \label{eq:optimal-B},
    \end{align}
    and they must satisfy $\langle\mc{B}^*-\mc{A}^*,\M\rangle \geq 0\ \forall\M\in\mathbb{F}_C$. In other words,
    \begin{align}
        \langle(d^2\Delta-\mc{P})\otimes(d^2\Delta-\mc{P}),\M\rangle \geq \frac{(1-F_c)^2}{F_c^2}\langle\mc{P}\otimes\mc{P},\M\rangle \quad\forall\M\in\mbb{F}_C. \label{eq:eq1}
    \end{align}
    Equivalently,
    \begin{align}
        \tr[(\mathbb{I}-P)\otimes(\mathbb{I}-P)J_\M] \geq \frac{(1-F_c)^2}{F_c^2}\tr(P\otimes P J_\M) \quad\forall\M\in\mathbb{F}_C.\label{eq:eq2}
    \end{align}     
    Now, we construct $\Theta_0$ as
    \begin{align}
        \Theta_0[\N] = C\left\{\tr[(P\otimes P) J_\N]\Phi + \frac{F_c}{1-F_c}\tr[(\mathbb{I}-P)\otimes(\mathbb{I}-P) J_\N]\Lambda\right\}.
    \end{align}
    This supermap is RNG because for any $\M\in\mathbb{F}_C$, 
    \begin{align}
        \Theta_0[\M] \propto A\Phi + (1-F_c)\Lambda \qquad \text{with}\quad A \leq F_c,
    \end{align}
    which is a free channel.

    Now let us consider the case where conditions 1 and 2(a) are met. When the resource theory is affine, the dual program of $\Omega(\N_F\otimes\N_F)$ becomes
    \begin{align}
        \sup\{\langle\mc{A}, \N_F\otimes\N_F\rangle:\ \tr\, \langle\mc{B},\N_F\otimes\N_F\rangle=1,\ \langle\mc{B}-\mc{A},\M\rangle = 0 \ \forall\M\in\mathbb{F}_C\}.
    \end{align}
    The proof thus proceeds analogously, except that we can immediately conclude that $x=\frac{(1-F_c)^2}{F_c^2(1-F)^2}$ from the fact that $\langle\mc{B}^*-\mc{A}^*,\M\rangle=0$, and the inequalities in Eqs.~\eqref{eq:eq1} and \eqref{eq:eq2} become equalities. We can thus conclude
    \begin{align}
        \Theta_0[\M] \propto F_c\Phi + (1-F_c)\Lambda \in \mbb{F}_C
    \end{align}

    It only remains to prove that $\Theta$ is universal for all states with $\langle\N,\Phi\rangle > F_c$. Since $J_\Phi\in\mathrm{supp}(P)$, we have $\tr(PJ_\N) \geq \tr(J_\Phi J_\N) = \langle\Phi,\N\rangle$. For any $\N_1$ and $\N_2$ such that $F_i \coloneqq \langle\N_i,\Phi\rangle > F_c$, we have
    \begin{align}
        \langle\N_{out},\Phi\rangle &= \frac{\tr PJ_{\N_1}\tr PJ_{\N_2}}{\tr PJ_{\N_1}\tr PJ_{\N_2} + \frac{F_c}{1-F_c}\tr((\mbb{I}-P)J_{\N_1})\tr((\mbb{I}-P)J_{\N_2})} \\
        &>\frac{F_1F_2}{F_1F_2 + \frac{F_c}{1-F_c}(1-F_1)(1-F_2)} \\
        &>\frac{F_1F_c}{F_1F_c+\frac{F_c}{1-F_c}(1-F_1)(1-F_c)} \\
        &= F_1,
    \end{align}
    where we have used the fact that $\frac{xy}{xy+c(1-x)(1-y)}$ is monotonically increasing in $x$ and $y$ for $c>0$ and $x,y\in[0,1]$. By the same argument, we also conclude that $\langle\N_{out},\Phi\rangle>F_2$. This demonstrates universality.
\end{proof}

\section{Proof of Thm. \ref{thm:uniqueness}}
\begin{proof}
    Note that an superinstrument is universal for all channels with $\langle\N,\Phi\rangle>F_c$ if and only if it is universal for all isotropic channels $\N_F$ with $F>F_c$. The latter is further equivalent to the superinstrument being universal for all isotropic channels $\N_F$ with $F\geq F_c$ (see Thm. S6 in the Supplementary Material of Ref.\cite{zang2025no}). Because we only need to focus on isotrpoic channels, we can without loss of generality take $\Theta_0$ to be of the form
    \begin{align}
        \Theta_0 = &a_1\tr[J_\N(P\otimes P)]\Phi + a_2\tr[J_\N(P\otimes(\mathbb{I}-P))]\Phi + a_3\tr[J_\N((\mathbb{I}-P)\otimes P)]\Phi \notag\\
        &+ a_4\tr[J_\N(\mathbb{I}-P)\otimes(\mathbb{I}-P)]\Phi + a_5\tr[J_\N((\mathbb{I}-P)\otimes(\mathbb{I}-P))]\Lambda
    \end{align}
    Since $\Theta_0$ must map the free isotropic channel $F_c\Phi+(1-F_c)\Lambda$ to a free channel, we deduce
    \begin{align}
        \frac{F_c^2a_1 + F_c(1-F_c)a_2 + (1-F_c)F_ca_3 + (1-F_c)^2a_4}{a_5(1-F_c)^2} \leq \frac{F_c}{1-F_c},
    \end{align}
    which implies 
    \begin{align}
        a_1 + \frac{1-F_c}{F_c}a_2 + \frac{1-F_c}{F_c}a_3 + \left(\frac{1-F_c}{F_c}\right)^2a_4 \leq \frac{1-F_c}{F_c}a_5. \label{eq:free-instrument-condition}
    \end{align}
    Now, suppose the two given channels are isotropic channels $\N_{F_1}$ and $\N_{F_2}$, then such a $\Theta_0$ yields an output channel such that
    \begin{align}
        \langle\N_{out},\Phi\rangle = \frac{a_1F_1F_2+a_2F_1(1-F_2)+a_3(1-F_1)F_2+a_4(1-F_1)(1-F_2)}{a_1F_1F_2+a_2F_1(1-F_2)+a_3(1-F_1)F_2+a_4(1-F_1)(1-F_2)+a_5(1-F_1)(1-F_2)}.
    \end{align}
    In particular, when $F_2=F_c$, 
    \begin{align}
        &1-\langle\N_{out},\Phi\rangle \notag\\
        &=\frac{a_5(1-F_1)(1-F_c)}{a_1F_1F_c+a_2F_1(1-F_c)+a_3(1-F_1)F_c+a_4(1-F_1)(1-F_c)+a_5(1-F_1)(1-F_c)} \\
        &= \frac{a_5(1-F_1)(1-F_c)}{F_1F_c\left[a_1 + \frac{1-F_c}{F_c}a_2 + \frac{1-F_1}{F_1}a_3 + \frac{(1-F_1)(1-F_c)}{F_1F_c}a_4\right]+a_5(1-F_1)(1-F_c)} \\
        &\geq \frac{a_5(1-F_1)(1-F_c)}{F_1F_c\left[a_1 + \frac{1-F_c}{F_c}a_2 + \frac{1-F_c}{F_c}a_3 + \frac{(1-F_c)^2}{F_c^2}a_4\right]+a_5(1-F_1)(1-F_c)} \\
        &\geq \frac{a_5(1-F_1)(1-F_c)}{a_5F_1(1-F_c)+a_5(1-F_1)(1-F_c)} = 1-F_1
    \end{align}
    where the first inequality is strict for $F_1>F_c$ if one of $a_3$ and $a_4$ is nonzero, and the second inequality is strict if Ineq.~\eqref{eq:free-instrument-condition} is strict. Therefore, if $\Theta$ is universal, then it must hold that $a_3=a_4=0$ and that Ineq.~\eqref{eq:free-instrument-condition} is an equality. Following the same argument for the case where $F_1=F_c$, we can conclude that $a_2$ must be zero as well. This concludes the proof.
\end{proof}

\section{Proof of Thm. \ref{thm:general-rt-nec-condition}}
\begin{proof}
    If there exists an isotropic channel $\N_F$ whose projective robustness is sub-multiplicative, it must be that $\mc{A}^*$ and $\mc{B}^*$ given in Eqs.~\eqref{eq:optimal-A} and \eqref{eq:optimal-B} are not feasible, which means that there exist $\M\in\mbb{F}_C$ such that
    \begin{align}
        \langle\mc{A}^*,\M\rangle - \langle\mc{B}^*,\M\rangle = \tr[(P\otimes P)J_\M] - \frac{F_c^2}{(1-F_c)^2}\tr[((\mbb{I}-P)\otimes(\mbb{I}-P))J_\M] > 0, 
    \end{align}
    For such an $\M$, $\Theta_0$ [Eq.~\eqref{eq:general-supermap}] will map $\M$ into an isotropic channel $F\Phi+(1-F)\Lambda$ with $F>F_c$, which is not free.

    By Lemma \ref{lemma:Rg=Rs}, if $R_g<R_s$, then $\{F:\N_F\in\mbb{F}_C\}=[F_{min},F_c]$ for some $F_{min}>0$. Furthermore, if $F_{min}=F_c$, then solving $R_s(\Phi)$ yields $R_s(\Phi)=\infty$. Therefore, $\{F:\N_F\in\mbb{F}_C\}=[F_{min},F_c]$ with $F_{min}\in(0,F_c)$. Then, 
    \begin{align}
        \Theta_0(\N_{F_{min}}^{\otimes 2}) = F_{min}^2\Phi + \frac{F_c}{1-F_c}(1-F_{min})^2\Lambda
    \end{align}
    Now,
    \begin{align}
        \frac{F_{min}^2}{\frac{F_c}{1-F_c}(1-F_{min})^2} < \frac{F_{min}^2}{\frac{F_{min}}{1-F_{min}}(1-F_{min})^2} = \frac{F_{min}}{1-F_{min}}.
    \end{align}
    This means that $\Theta_0(\N_{F_{min}}^{\otimes 2})$ is proportional to a channel that is not free. Therefore, $\Theta_0$ is not RNG.
\end{proof}

\end{document}